\documentclass[11pt,letter]{article}
\usepackage{LEviaGrover-macro}
\usepackage{times}
\usepackage{ifthen}
\usepackage{amstext,amsgen,latexsym,amsmath}
\usepackage{amssymb,amsfonts}

\usepackage[left=1in,right=1in,top=1in,bottom=1in]{geometry}
\newenvironment{step}
  {
    \begin{enumerate}

  }
  {\end{enumerate}}

\newenvironment{algorithm*}[1]
  {
    \begin{center}
      \hrulefill\\
      \textbf{#1}
  }
  {
    
\vspace{-1\baselineskip}
    \hrulefill
    \end{center}
  }

\newenvironment{protocol*}[1]
  {
    \begin{center}
      \hrulefill\\
      \textbf{#1}
  }
  {
    \vspace{-1\baselineskip}
    \hrulefill
    \end{center}
  }

\renewcommand{\thefootnote}{\fnsymbol{footnote}}

\newcommand{\reg}[1]{{\sf#1}}
\newcommand{\cat}[2]{\mathrm{CAT}_{#1}(#2)}

\begin{document}
\begin{titlepage}

\sloppy

\title{
Computing on Anonymous Quantum Network
  }

\author{
{\large \hspace*{-1ex} $\text{Hirotada Kobayashi}^{\ast\dagger}$
 \hspace*{-1ex}}\\
{\tt hirotada@nii.ac.jp}
\and
{\large \hspace*{-1ex} $\text{Keiji Matsumoto}^{\ast\dagger}$
 \hspace*{-1ex}}\\
{\tt keiji@nii.ac.jp}
\and
{\large \hspace*{-1ex} $\text{Seiichiro Tani}^{\dagger\ddagger}$
 \hspace*{-1ex}}\\
{\tt tani@theory.brl.ntt.co.jp}
}

\date{}

\maketitle
\thispagestyle{empty}
\pagestyle{plain}
\vspace*{-5mm}
\begin{center}
  ${}^{\ast}$Principles of Informatics Research Division,  National Institute of Informatics.\\
  ${}^{\dagger}$Quantum Computation and Information Project, ERATO-SORST, JST.\\
  ${}^{\ddagger}$NTT Communication Science Laboratories, NTT Corporation.\\
\end{center}
\begin{abstract}
  This paper considers distributed computing on an \emph{anonymous
    quantum network}, a network in which no party has a unique
  identifier and quantum communication and computation are available.
  It is proved that the leader election problem can exactly (i.e.,
  without error in bounded time) be solved with at most the same
  complexity up to a constant factor as that of exactly
  computing symmetric functions (without intermediate measurements for a distributed and superposed
  input), if the number of parties is given to every party.
A corollary of this result is a more efficient quantum leader election
algorithm than existing ones: the new quantum algorithm runs in $O(n)$ rounds with bit complexity $O(mn^2)$,
on an anonymous quantum network with $n$ parties and $m$ communication links.
It follows that
all Boolean functions computable on a \emph{non}-anonymous quantum
network can be computed with the same order of complexity as the
quantum leader election algorithm on an anonymous quantum network.
This gives the first quantum algorithm that exactly computes any
computable Boolean function with round complexity $O(n)$ and with
smaller bit complexity than that of existing classical algorithms in
the worst case over all (computable) Boolean functions and network topologies.
More generally,
any $n$-qubit
state
can be shared with that complexity on an anonymous quantum network with $n$ parties.
This paper also examines an important special case:
the problem of sharing an $n$-partite GHZ state
among $n$ parties on an anonymous quantum network.
It is proved that there exists a quantum algorithm that exactly solves 
this problem with rounds linear in
the number of parties
with a \emph{constant}-sized gate set.

\end{abstract}

%%% Local Variables: 
%%% mode: latex
%%% TeX-master: "KobMatTan10"
%%% End: 

\bigskip
\centerline{{\bf Keywords}: quantum computing, distributed computing, leader election.}
\end{titlepage}
\maketitle
\pagestyle{plain}
\sloppy

\renewcommand{\thefootnote}{\arabic{footnote}}

\section{Introduction}
\subsection{Background}
Distributed computing algorithms often 
depend on the assumption that messages sent by distinct parties
can be distinguished, which is justified
if every party has a unique identifier.
A more general case without this assumption 
is an \emph{anonymous network},
i.e., a network where no party has a unique identifier.
Computing on anonymous networks
was first considered with the leader election problem in Ref.~\cite{Ang80STOC} and has been further investigated in the literature (e.g., Refs.~\cite{ItaRod81FOCS,ItaRod90InfoComp,afek-matias94,KraKriBer94InfoComp,
YamKam96IEEETPDS-1,YamKam96IEEETPDS-2}).
This setting makes it significantly hard or even impossible to solve
some distributed computing problems that are easy to solve on a non-anonymous network.

The leader election problem is the problem of  
electing a unique leader from among distributed parties
and it is a fundamental problem:
Once it is solved, 
the leader can gather all distributed input and locally solve any distributed computing problem
(except cryptographic or fault-tolerant problems) (e.g., Ref.~\cite{Lyn96Book}).
However, it was proved in Refs.~\cite{Ang80STOC,YamKam88PODC,YamKam96IEEETPDS-1}
that
no classical algorithm 
can exactly solve the
leader election problem on anonymous networks for a
certain broad class of network topologies,
such as rings and a certain family of regular graphs,
even if the network
topology (and thus the number of parties) is known to each party prior
to algorithm invocation.
Here,  an algorithm is said to \emph{exactly solve} a problem
  if it solves the problem without fail in bounded time.
Thus, many other problems
have also been 
studied
to clarify their solvability on anonymous networks:
some were shown
to be exactly solvable (for certain families of graphs) and others were not ~\cite{YamKam96IEEETPDS-1}. 
For instance, any symmetric Boolean function can be
computed on an anonymous network of any unknown topology, 
if the number of parties is given to each party \cite{YamKam88PODC,YamKam96IEEETPDS-1,KraKriBer94InfoComp};  
in particular, efficient algorithms are known for
various regular graphs (e.g., Refs.~\cite{AttSniWar88JACM, KraKri92SPDP,KraKriBer94InfoComp,KraKri97JALG}).

Surprisingly, the situation is quite different on \emph{quantum networks}, i.e., networks in which
quantum computation and communication are available.
It was
proved by the present authors in Ref.~\cite{TanKobMat05STACS}%
\footnote{
A nice survey of this article is found in Ref.~\cite{DenPan06SIGACT}
} 
that the leader election
problem can exactly be solved on an anonymous quantum network of any unknown
topology,
if 
the number of parties
is given to every party.  
This implies that quantum power substantially changes the computability of the leader election problem
on anonymous networks.

Our questions are then as follows: How powerful is quantum information
for solving distributed computing tasks?  Does quantum power change
the hardness relation among distributed computing problems (e.g.,
problem A is harder than problem B in the classical setting, while
they have similar hardness in the quantum setting)?  We give an answer
to these questions by comparing the leader election problem with
computing symmetric functions, well-known problems that can be solved
even on an anonymous classical network.  As a corollary, we provide a more
efficient quantum leader election algorithm than existing ones.  For
every Boolean function computable on a \emph{non}-anonymous quantum
network (a quantum network in which every party has a unique
identifier), this yields a quantum algorithm that computes it on an
anonymous quantum network with the same order of complexity as the
quantum leader election algorithm.  In distributed quantum
computing, sharing a quantum state among parties
is also a fundamental problem.  The above algorithm of computing Boolean functions
actually solves the problem of $n$ parties sharing any quantum state.
We also examine an important special case: the problem of
sharing an $n$-partite GHZ state among $n$ parties on an anonymous
quantum network, called the GHZ-state sharing problem.

\subsection{Main Results}
Hereafter, we assume that the underlying graphs of networks are undirected
and that no faults exist on networks.
\subsubsection{Quantum Leader Election}
Our first result shows that
the leader election problem is not harder than computing symmetric functions
on anonymous quantum networks.
Let $n$ be the number of parties
and $H_k\colon \set{0,1}^n\rightarrow\set{\true,\false}$ be the function over
distributed $n$ bits, which is $\true$
if and only if the Hamming weight, i.e., the sum, of the $n$ bits is
$k$. 
Let $\calH_k$ be any quantum algorithm that exactly 
computes $H_k$ without intermediate measurements%
\footnote{ The condition ``without intermediate measurements'' is
  required for clear definition. 
It is easy to convert any quantum
  algorithm involving intermediate measurements into a quantum
  algorithm not involving them by postponing all the measurements.
  However, this conversion may increase the bit complexity and, thus,
  these two kinds of algorithms should be considered separately.
For instance, consider a quantum algorithm involving intermediate measurements
that uses a different subset of communication links for each intermediate measurement results,
in which case the algorithm uses the union of the subsets when postponing
the intermediate measurements.
}
on an anonymous quantum network, and let $Q^{\operatorname{rnd}}(\calH_k)$ and
$Q^{\operatorname{bit}}(\calH_k)$ be the worst-case round and bit complexities of
$\calH_k$ over all possible quantum states as input.\footnote{ Note
  that inputs are not limited to classical inputs when we consider the
  complexities of quantum algorithms that compute classical
  functions. For instance, $\calH_k$ can take a state of the form
  $\sum_{\vec{x}} \alpha_{\vec{x}} \ket{\vec{x}}$ as input, where
  $\alpha_{\vec{x}} \in \Complex$ and $\vec{x} \in \{0,1\}^n$. This
  takes it into account that the algorithm may use smaller amounts of
  communications when restricted to classical inputs. For instance,
  for each classical input, the algorithm may use a different subset
  of communication links, in which case it uses the union of these
  subsets of communication links when the input is a superposition of
  such classical inputs, and may result in the increase of
  communication complexities. Hence, the correct way of defining
  complexities of quantum algorithms that are used as subroutines is
  taking the maximum over all possible inputs, including quantum
  inputs.}
\begin{theorem}
\label{th:LE}
If the number $n$ of parties is provided to each party,
the leader election problem can exactly be solved 
in 
$O(Q^{\operatorname{rnd}}(\calH_0)+Q^{\operatorname{rnd}}(\calH_1))$ 
rounds
with bit complexity
$O(Q^{\operatorname{bit}}(\calH_0)+Q^{\operatorname{bit}}(\calH_1))$
on  an anonymous 
quantum network of any unknown topology.
\end{theorem}
This is the first non-trivial characterization of the complexity
of leader election relative to computing Boolean functions
on anonymous 
quantum networks.  This
does not have a classical counterpart, since,
for some network topologies (e.g., rings),
symmetric Boolean
functions can exactly be computed \cite{YamKam88PODC,YamKam96IEEETPDS-1,KraKriBer94InfoComp} but 
a unique leader cannot exactly be
elected \cite{YamKam96IEEETPDS-1}.  
In fact, any symmetric function can exactly be computed 
on an anonymous classical network of any unknown topology
(and thus, on an anonymous quantum network). 
Therefore, Theorem~\ref{th:LE} subsumes the computability result
in Ref.~\cite{TanKobMat05STACS} 
that the leader election problem can
exactly be solved on an anonymous quantum network.
Our second result is that computing $H_1$ is reducible to computing $H_0$. 
\begin{theorem}
\label{th:H1}
If the number $n$ of parties is provided to each party,
$H_1$ can exactly 
be computed 
without intermediate measurements
for any possible quantum states as input
in 
$O(Q^{\operatorname{rnd}}(\calH_0))$ rounds
with bit complexity 
$O(n\cdot Q^{\operatorname{bit}}(\calH_0))$
on  an anonymous 
quantum network of any unknown topology.
\end{theorem}
Theorem~\ref{th:LE} together with Theorem~\ref{th:H1} implies that the complexity of the leader election problem
is characterized by that of computing $H_0$.
This would be helpful in intuitively understanding the hardness of the leader election problem on an anonymous quantum network,
since computing $H_0$ can be interpreted as just a simple problem of checking if all parties
have the same value.

Since Theorem~\ref{th:LE} (Theorem~\ref{th:H1}) is
proved by quantumly reducing the leader election problem (resp.~computing $H_1$) to computing
$H_0$ and $H_1$ (resp.~computing $H_0$), the theorems provide ways of developing quantum leader election algorithms
by 
plugging in 
algorithms that compute $H_0$ (and $H_1$).
Since there is a classical
algorithm that exactly computes $H_0$ in $O(n)$ rounds with bit
complexity $O(mn)$ 
for the number $m$ of edges of the underlying graph
(e.g., Ref.~\cite{KraKriBer94InfoComp})
and it can be converted into a quantum algorithm with the same complexity up to a constant factor,
Theorem~\ref{th:LE} together with Theorem~\ref{th:H1}
yields a quantum leader election
algorithm.
\begin{corollary}
\label{cr:LE}
The leader election problem can exactly be solved in
$O(n)$ rounds with bit complexity $O(mn^2)$
on an anonymous 
%undirected 
quantum network for any unknown topology,
if the number $n$ of parties is given to every party,
where $m$ is the number of edges of the underlying graph.
\end{corollary}
This leader election algorithm has better round and bit complexity than
existing algorithms --- the two quantum algorithms given in
Ref.~\cite{TanKobMat05STACS} have the round [bit] complexity of
$O(n^2)$ [$O(mn^2)$] and $O(n \log n)$ [$O(mn^4 \log n)$],
respectively.
Actually, the proofs of Theorems ~\ref{th:LE} and \ref{th:H1} 
can be carried over asynchronous networks in a straightforward manner.
Thus, the theorems hold for asynchronous networks.

\subsubsection{Quantum State Sharing}
Once a unique leader is elected, 
it is possible to 
construct a spanning tree and
assign a
unique identifier drawn from $\set{1,\dots ,n}$ to each party
with the same order of complexity as that of electing a unique leader on anonymous quantum networks.%
\footnote{The problem of assigning unique identifiers drawn from a small domain
has been widely studied even for non-anonymous networks since the length of each identifier
has a great influence over the bit complexity  of many problems.}
Then, the leader can recognize the underlying graph 
by gathering along the spanning tree the adjacency matrices of subgraphs with a unique identifier on each node.
This implies that, 
if every party $i$ is given a bit $x_i$ as input,
a unique leader (who is elected by the leader election algorithm)
can compute any Boolean function $f(x_1,\dots ,x_n)$ that depends on 
the underlying graph $G$ with node label $x_i$s
(and send the function value to every party along the spanning tree).
Here, the index $i$ of each party is introduced just for explanation,
and it is not necessarily the same as the identifier assigned by the leader to the party having $x_i$.
An example of $f$ is a majority function that is $\true$ if and only if the sum over all $x_i$'s is more than $n/2$.
Another example is a function that is $\true$ if and only if there is a cycle in which each node $i$ has input $x_i=1$.
Similarly, if each party is given a qubit as node label so that the
$n$ parties share some $n$-qubit state $\xi$, the leader can generate
any quantum state $\rho$ computable from $\xi$ and the underlying
graph $G$.
\begin{corollary}
\label{cr:Boolean}
Suppose that every party $i$ is given the number $n$ of parties and a
qubit as node label so that the $n$ parties share some $n$-qubit state
$\xi$. Let $\rho$ be any $n$-qubit quantum state computable from $\xi$
and the underlying graph.
Then, state $\rho$ can exactly be shared among the $n$
parties in $O(n)$ rounds with bit complexity $O(mn^2)$ on an anonymous
quantum network,
where $m$ is the number of edges of the underlying graph. 
A special case of $f$ is a Boolean function that is
determined by the underlying graph in which each node $i$ is labeled
with a bit $x_i$.  If every party $i$ is given $n$ and $x_i$, function
$f$ can exactly be computed in $O(n)$ rounds with bit complexity
$O(mn^2)$ on an anonymous quantum network.
\end{corollary}
This gives the first quantum algorithm that exactly computes any
computable Boolean function with round complexity $O(n)$ and with
smaller bit complexity than that of existing classical 
algorithms~\cite{YamKam96IEEETPDS-1,KraKriBer94InfoComp,TanKobMat05STACS}
in
the worst case over all (computable) Boolean functions and network
topologies.

%\myspace
\paragraph{GHZ-State Sharing} 
From the viewpoint of quantum information,
our leader election algorithm exactly solves the problem of sharing an $n$-partite $W$-state
(e.g., a state $(\ket{100}\+\ket{010}\+\ket{001})/\sqrt{3}$ for the three-party case).
As described above, this essentially solves the more general problem
of sharing an $n$-qubit state $\rho$.
We are then interested in whether a certain non-trivial $\rho$ can be shared with less computational resources
than a $W$-state. Specifically, we focus on the number of distinct quantum gates required to share $\rho$, 
since, for the leader election problem, all known exact algorithms (including ours) require quantum gates
that depend on the number $n$ of parties. 

Among non-trivial quantum states other than $W$-states, 
an $n$-partite GHZ state would be one of the most interesting quantum states,
since it would be a useful resource for
quantum computing and communication.
We give exact quantum algorithms that solve, with a
constant-sized gate set, the problem of sharing an $n$-partite GHZ state (\emph{or} an $n$-partite cat
state) $(\ket{0}^{\otimes n}\+\ket{1}^{\otimes n})/\sqrt{2}$ with qubits, and the problem
of sharing an $n$-partite generalized-GHZ state $(\ket{0}^{\otimes n}\+\cdots \+\ket{k-1}^{\otimes n})/\sqrt{k}$ 
with $k$-level qudits for a constant integer $k\ge 2$, among $n$
parties on an anonymous quantum network.  
 We call this problem the
 \emph{GHZ-state sharing problem}.
Notice that $k$-level qudits are physically realizable \cite{NieChu00Book} and are just qubits for $k=2$.
Let $F_k$ be a function such that $F_k(x_1,\ldots, x_n)=\sum _{i=1}^{n} x_i \pmod{k}$ for distributed inputs $x_i\in \set{0,\dots ,k-1}$.
Let $\calF_k$ be any quantum algorithm that exactly 
computes $F_k$ 
without intermediate measurements
on an anonymous network,
and
let $Q^{\operatorname{rnd}}(\calF_k)$ and $Q^{\operatorname{bit}}(\calF_k)$ be the worst-case round and bit complexities, respectively, 
of 
$\calF_k$ 
over all possible quantum states as input.
\begin{theorem}
\label{th:QuantumConsensus}
If every party is given the number $n$ of party and an integer $k\ge 2$,
the GHZ-state sharing problem
can exactly be solved on an anonymous quantum network
in 
$O(Q^{\operatorname{rnd}}(\calF_k))$ 
%$2Q^{\operatorname{rnd}}(\calF_k)$ 
rounds with bit complexity 
$O(Q^{\operatorname{bit}}(\calF_k))$.
%$2k\cdot Q^{\operatorname{bit}}(\calF_k)$,
%where $d(k)$ is the number of divisors (including ``1'' and ``k'' ) of $k$.
Moreover, every party uses only a constant-sized gate set 
to perform all operations for any integer constant $k\ge 2$, if an algorithm $\calF_k$ is given as a black box.
\end{theorem}

For every integer constant $k\ge 2$, there is an algorithm that exactly 
and reversibly 
computes $F_k$  
for any possible quantum state as input
%for any superposed input 
in $O(n)$ rounds with bit
complexity $O(mn^4 \log n)$ on
an anonymous classical/quantum network of any unknown topology
\cite{TanKobMat05STACS}.  Therefore, the theorem implies that there
exists a quantum algorithm that exactly solves the GHZ-state sharing problem
with a constant-sized gate set for any constant $k\ge 2$.  For
$k=2$, we have the following corollary.
\begin{corollary}
The GHZ-state sharing problem with $k = 2$
can exactly be solved
on an anonymous quantum network
for any number $n$ of parties
with a gate set  that can perfectly implement the Hadamard transformation
and any classical reversible transformations.
In particular, the problem can exactly be solved
with either the Shor basis or the gate set
consisting of the Hadamard gate, the CNOT gate, and the Toffoli gate.
\end{corollary}
If much more rounds are allowed, there exists a more bit-efficient algorithm
that exactly solves the GHZ-state sharing problem in $O(n^2)$ rounds with bit complexity $O(mn^2)$ by using only a constant-sized gate set for any $n$. The algorithm is obtained
by modifying Algorithm I in Ref.~\cite{TanKobMat05STACS}.

%\myspace
\subsection{Related Work}
Refs.~\cite{PalSinKum03ARXIV,DHoPan06QIC} have dealt with
the leader election and GHZ-state sharing problems
in a different setting where pre-shared entanglement is assumed
but only classical communication is allowed.
The relation between several network models that differ in available quantum resources
is discussed in Ref.~\cite{GavKosMar09ARXIV}.
\subsection{Organization}
Section~2 describes the network model, and some tools and notations used in the paper.
Sections~3 and 4 prove Theorems~\ref{th:LE} and \ref{th:H1}. Section ~5 then gives a quantum leader election algorithm
as a corollary of the theorems. Section~6 considers the problems of computing Boolean functions and sharing a quantum state.
Section~7 presents a quantum algorithm for the GHZ-state sharing problem.

%%% Local Variables: 
%%% mode: latex
%%% TeX-master: "KobMatTan10"
%%% End: 

\section{Preliminaries}

\subsection{\bf Distributed Computing}
\label{subsec:model}
\paragraph{The Network Model:}
A classical \emph{network}
is composed of
multiple parties
and bidirectional classical communication links connecting parties.
In a quantum network,
every party can perform quantum computation and communication,
and each adjacent pair of parties
has a bidirectional quantum communication link between them
(we do not assume any prior shared entanglement).
When the
parties and links are regarded as nodes and edges, respectively,
the topology of the network is expressed by a connected undirected
graph. 
We denote by $\calG_n$ the set of all $n$-node connected undirected graphs 
with no multiple edges and no self-loops.
In what follows,
we may identify each party/link with its corresponding node/edge
in the underlying graph for the system,
provided that doing so is not confusing.
%if it is not confusing.
Every party has \emph{ports} 
corresponding one-to-one to communication links incident to the party.
Every port of party $l$ has a unique label $i$, $(1 \leq i \leq d_l)$,
where $d_l$ is the number of parties adjacent to $l$.
More formally,
the underlying graph ${G=(V,E)}$ has a \emph{port numbering} \cite{YamKam96IEEETPDS-1},
which is a set $\sigma$ of functions $\{ \sigma[v] \colon v\in V\}$
such that, for each node $v$ of degree $d_v$,
$\sigma[v]$ is a bijection from the set of edges incident to $v$
to $\{ 1, 2, \ldots, d_v \}$.
It is stressed that
each function $\sigma[v]$ may be defined independently of any other $\sigma [v']$.
In our model,
each party knows the number of his ports and
the party can appropriately choose one of his ports
whenever he transmits or receives a message.

Initially, every party $l$ has local information $I_l$,
the information that only party $l$ knows,
such as his local state and the number of his adjacent parties,
and global information $I_G$,
the information shared by all parties (if it exists),
such as the number of parties in the system
(there may be some information shared by not all parties, but 
it is not necessary to consider such a situation
when defining anonymous networks).
Every party $l$ runs the same algorithm,
which is given local and global informations, $I_l$ and $I_G$,  as its arguments.
If all parties have the same local information except
for the number of ports they have,
the system and the parties in the system are said to be \emph{anonymous}.
For instance, if the underlying graph of an anonymous network is regular,
this is essentially equivalent
to the situation in which every party has the same identifier
(since we can regard the local information $I_l$ of each party $l$ as his identifier).
This paper deals with only anonymous networks, but
may refer to a party with its index (e.g., party $i$)
only for the purpose of simple description.

A network is either
\emph{synchronous} or \emph{asynchronous}. In the synchronous case, message passing is performed synchronously.
The unit interval of synchronization is called a \textit{round}.
Following the approach in Ref.~\cite{Lyn96Book},
one round consists of
the following two sequential steps, where 
we assume that  two (probabilistic) procedures
that generate messages and change local states are defined in the algorithm
invoked by each party:
(1) each party changes his local state according to a procedure
that takes his current local state and the incoming messages as input, and then removes the
  messages from his ports;
(2) each party then prepares messages and decides the ports
through which the messages should be sent
by using the other 
procedure that takes his current local state as input,
and finally sends the messages via the ports.
Notice that, in the quantum setting, the 
two procedures 
are physically realizable operators.
A network that is not synchronous is asynchronous.
In asynchronous networks, the number of rounds required by an algorithm
is defined by convention as the length of the longest chains of messages
sent during the execution of the algorithm.

This paper focuses on the required number of rounds as a complexity
measure (called \emph{round complexity}). This is often used as an
approximation of time complexity, which includes the time taken by
local operations as well as that taken by message exchanges.
Another complexity measure we use is
bit complexity, which is the number of bits, including qubits,  communicated over all communication links. 
In this paper, we do not assume any faulty party and communication link.

\subsection{Leader Election Problem in Anonymous Networks}
The leader election problem is formally defined as follows.
\begin{definition}[\boldmath{$n$-party leader election problem ($\LE_n$)}]
Suppose that there is an $n$-party network whose
underlying graph is in $\calG_n$, and that each party
$i\in \{ 1,2,\ldots, n\}$ in the
network has a variable
$y_i$ initialized to $1$.
Create the situation in which ${y_k=1}$ for a certain
$k\in \{ 1,2,\ldots, n\}$
and ${y_i=0}$ for every $i$ in the rest 
$\{ 1,2,\ldots, n\}\setminus \{k\}$.
\end{definition}
This paper considers $\LE_n$ on an anonymous network
(when each party $i$ has his own unique identifier, i.e.,  
$I_i\neq I_j$ for all distinct $i,j\in \set{1,\dots ,n}$,
$\LE_n$ can 
deterministically be solved in  $\Theta(n)$ rounds in both synchronous and 
asynchronous cases~\cite{Lyn96Book}).

The leader election problem on an anonymous network was first 
investigated by
Angluin~\cite{Ang80STOC}. 
Subsequently,
Yamashita and
Kameda~\cite{YamKam96IEEETPDS-1} 
gave a necessary and sufficient condition on network topologies
under which $\LE_n$ can exactly be solved for given $n$.
Their result implies that $\LE_n$ cannot exactly be solved 
for a broad class of graphs, including
rings, complete graphs, and certain families of regular graphs. 
Interested readers should consult Refs.~\cite{afek-matias94,yamashita-kameda99}
and the references in them for detailed information about the leader election problem on anonymous networks.

\subsection{Quantum Computing}
We assume that readers have some basic knowledge of quantum computing
introduced in standard textbooks \cite{NieChu00Book,KitSheVya02Book}.
The following well-known theorem is called ``exact quantum amplitude amplification'',  which will be used repeatedly.
\begin{theorem}[\protect\cite{BraHoyMosTap02AMS,chi-kim98}]
\label{th: quantum amplitude amplification}
  Let $\calA$ be any quantum algorithm that searches
for $z\in\set{0,1}^n$ such that $\chi(z)=\true$
without using measurements, where
$\chi(z)\in \{\true,\false\}$ is any Boolean function.
Suppose that
\[
\ket{\Psi}=\calA\ket{0}^{\otimes n}=\sum _{z}\alpha_z \ket{z}
\]
for orthonormal basis $\set{\ket{z}}_{
z \in \{0,1\}^n
}$.
Let $Q(\calA,\chi,\phi,\theta)$
be an operator 
\[-\calA F _{0}(\phi)\calA^{-1}F _{\chi}(\theta), \] 
where
$F _{\chi}(\theta)$ multiplies $\ket{z}$ by a factor of 
$e^{i\theta}$ if $\chi(z)=\true$, and
$F _{0}(\phi)$ multiplies $\ket{z}$ by a factor of 
$e^{i\phi}$ if $z=0\cdots 0$.

If
  the initial success probability $a=\sum_{z\colon \chi(z)=\true}\abs{\alpha_z}^2$
of $\calA$ is exactly known and at least $1/4$, then
$$
Q(\calA,\chi,\phi_a,\theta_a)\ket{\Psi}
=\frac{1}{\sqrt{a}}\sum _{z\colon \chi(z)=\true}\alpha_z \ket{z}
$$
for some values
$\phi_a$ and $\theta_a$ $(0\leq \phi_a,\theta_a\leq 2\pi)$
computable from $a$.
 \end{theorem}
\subsection{Notations}
A Boolean function $f\colon \set{0,1}^n\rightarrow \set{\true,\false}$
depending on $n$ variables, 
$x_1,\dots ,x_n$
with $x_i\in \set{0,1}$,
is said to be \emph{symmetric} if $f$ is determined by
the Hamming weight of $\vec{x}=(x_1,\dots ,x_n)$, i.e., $\abs{\vec{x}}=\sum _{i=1}^n x_i$.
In particular, symmetric function 
$H_k\colon \set{0,1}^n\rightarrow \set{\true,\false}$ 
is defined as 
$H_k(\vec{x})=\true$
if and only if 
$\abs{\vec{x}}$ is $k$.
We say that $n$ parties \emph{exactly compute a Boolean function} $f\colon \set{0,1}^n\rightarrow \set{\true,\false}$
if every party $i$ 
has variables $y_i$ (initialized to ``$\true$'')
and $x_i\in \set{0,1}$ before computation, and set
$y_i$ to $f(\vec{x})$ with certainty after computation.
If a quantum algorithm 
exactly 
computes $f$ 
without intermediate measurements
on an anonymous quantum network, we say that the algorithm is an \emph{$f$-algorithm}.

In general, an $f$-algorithm
transforms (with
ancilla qubits) an input state $\left[ \bigotimes_{i=1}^n
  (\ket{x_i}\ket{\true}) \right] \otimes \ket{0}$ into $\left[
  \bigotimes_{i=1}^n (\ket{x_i}\ket{f(\vec{x})}) \right] \otimes
\ket{g_{\vec{x}}}$, for any $\vec{x} = (x_1, \ldots, x_n) \in
\{0,1\}^n$, where $\ket{g_{\vec{x}}}$ is ``garbage'' left after
computing $f(\vec{x})$. For the algorithms over networks with
bidirectional communication links, any $f$-algorithms
are reversible. Hence we can totally remove the ``garbage'' by
standard garbage-erasing technique, as the $f$-algorithm exactly and reversibly
computes $f$. Putting everything together, we may assume without loss
of generality (at the cost of doubling each complexity) that any
$f$-algorithm transforms an input state
\[
\sum_{\vec{x} \in \{0,1\}^n}
 \alpha_{\vec{x}}
 \bigotimes_{i=1}^n
  (\ket{x_i}\ket{\true})
\]
into
\[
\sum_{\vec{x} \in \{0,1\}^n}
 \alpha_{\vec{x}}
 \bigotimes_{i=1}^n
  (\ket{x_i}\ket{f(\vec{x})}),
\]
for any $\alpha_{\vec{x}} \in \Complex$ with $\sum_{\vec{x} \in \{0,1\}^n} \abs{\alpha_{\vec{x}}}^2 =1$, where $\vec{x} = (x_1, \ldots, x_n)$.
Similarly, for the more general function $f\colon X^n\to Y$ depending on distributed $n$ variables
$(x_1,\dots ,x_n)$ with $x_i\in X$,
we say that a quantum algorithm is an $f$-algorithm,
if the algorithm 
exactly 
computes $f$ 
without intermediate measurements 
on an anonymous quantum network.
For an $f$-algorithm $\calF$ on an anonymous quantum network with the underlying graph $G \in \calG_n$,
we denote by 
 $Q_G^{\operatorname{bit}}(\calF)$ and 
 $Q_G^{\operatorname{rnd}}(\calF)$ 
the worst-case bit and round complexities, respectively, 
of 
$\calF$
over all possible quantum states given as input.
For simplicity, we may write 
$Q^{\operatorname{bit}}(\calF)$ and 
$Q^{\operatorname{rnd}}(\calF)$
if $G$ is clear from context.

%%% Local Variables: 
%%% mode: latex
%%% TeX-master: "KobMatTan10"
%%% End: 

\section{Proof of Theorem~\ref{th:LE}}
\label{subsec:ProofOfLE}
\subsection{Basic Idea}
 Initially, every party is eligible to be the leader
and is given the number $n$ of parties as input.
 Every party flips a coin that gives heads with probability $1/n$
 and
tails with $1-1/n$. 
If exactly one party sees heads, the party becomes  a unique leader.
The probability of this successful case is given by
\[
s(n)
=
  {n \choose 1}
  \cdot
  \frac{1}{n}
  \cdot
  \left( \frac{n-1}{n} \right)^{n-1}
=
\left(
  1 - \frac{1}{n}
\right)^{n-1}
>
\frac{1}{e}
>
\frac{1}{4}.
\]
We shall amplify the probability of this case to one
by applying the exact quantum amplitude amplification 
in Theorem~\ref{th: quantum amplitude amplification}.
To do this, we use an $H_1$-algorithm  in a black-box manner
to check (in $F_{\chi}(\theta_{s(n)})$) whether or not a run of the above randomized algorithm results in 
the successful case,
and use an $H_0$-algorithm in a black-box manner to realize the diffusion operator
(more strictly, $F_0(\phi_{s(n)})$).
In other words, we shall quantumly reduce the leader election problem to computing $H_0$ and $H_1$.
In our algorithm, all communication  is performed 
for computing  $H_0$, $H_1$ and their inversions. 
The non-trivial part is how
to implement $F_{\chi}(\theta_{s(n)})$ and $F_0(\phi_{s(n)})$ in a distributed way
on an anonymous network, where $s(n)=(1 - 1/n)^{n-1}$, since every party must run
the same algorithm.

\subsection{The Algorithm}
Before describing the algorithm,
we introduce the concept of
\emph{solving} and \emph{unsolving} strings.
Suppose that each party $i$ has a bit $x_i$,
i.e.,
the $n$ parties share $n$-bit string ${\vec{x} = (x_1, x_2 ,\ldots, x_n)}$.
A string $\vec{x}$ is said to be \emph{solving}
if $\vec{x}$ has Hamming weight one.
Otherwise, $\vec{x}$ is said to be \emph{unsolving}.
We also say that an $n$-qubit pure state
${\ket{\psi} = \sum_{\vec{x} \in \{0, 1\}^n} \alpha_{\vec{x}} \ket{\vec{x}}}$
shared by the $n$ parties is \emph{solving (unsolving)}
if ${\alpha_{\vec{x}} \neq 0}$ only for $\vec{x}$ that is solving (unsolving).

Fix an $H_0$-algorithm and an $H_1$-algorithm, which we are allowed to use in a black-box manner.

\paragraph{Base algorithm $\calA$:} Let $\bfA$ be the two-by-two unitary matrix defined by
\[
\bfA
=
\frac{1}{\sqrt{n}}
\begin{pmatrix}
\sqrt{n-1} & 1\\
1 & - \sqrt{n-1}
\end{pmatrix}.
\]
At the beginning of the algorithm,
each party prepares three single-qubit quantum registers
$\reg{R}$, $\reg{S}$, and $\reg{S}^\prime$,
where the qubit in 
$\reg{R}$ is initialized to $\ket{0}$,
the qubits in $\reg{S}$ and $\reg{S}^\prime$ are initialized to $\ket{\mbox{``$\true$''}}$
(the qubits in $\reg{S}$ and $\reg{S}^\prime$ will be used as ancillary qubits when performing phase-shift operations
on the qubit in $\reg{R}$).
First, each party applies $\bfA$ to the qubit in $\reg{R}$
to generate the quantum state
${
\ket{\psi}
=\bfA \ket{0}=
\sqrt{1-\frac{1}{n}}\,\ket{0} + \sqrt{\frac{1}{n}}\,\ket{1}
}$.
Equivalently,
all $n$ parties share the $n$-qubit quantum state
\[
{
\ket{\Psi}
=
\ket{\psi}^{\otimes n}
=
\left(
\sqrt{1-\frac{1}{n}}\,\ket{0} + \sqrt{\frac{1}{n}}\,\ket{1}
\right)^{\otimes n}
}
\]
in their $\reg{R}$'s.
Let ${S_n = \set{\vec{x} \in \{0,1\}^n\colon \text{$\vec{x}$ is solving}}}$
be the set of solving strings of length $n$,
and let
${
\ket{\Psi_{\solving}}
=
\frac{1}{\sqrt{n}} \sum_{\vec{x} \in S_n} \ket{\vec{x}}
}$
be the quantum state which is the uniform superposition of
solving strings of length $n$.
Notice that
$\ket{\Psi}$ is a superposition of
the solving state $\ket{\Psi_{\solving}}$
and some unsolving state $\ket{\Psi_{\unsolving}}$:
\[
\ket{\Psi}
=\alpha_{\solving}\ket{\Psi_{\solving}}+ \alpha_{\unsolving}\ket{\Psi_{\unsolving}}.
\]
The amplitude $\alpha_{\solving}$ of $\ket{\Psi_{\solving}}$
is given by $\alpha_{\solving}=\sqrt{s(n)}> 1/2$.

\paragraph{Exact amplitude amplification:} Now the task for the $n$ parties is
to amplify the amplitude of $\ket{\Psi_{\solving}}$
to one via exact amplitude amplification,
which involves 
one run of $-\calA F _{0}(\phi_a)\calA^{-1}F _{\chi}(\theta_a)$ for $\calA=\bfA^{\otimes n}$
since the initial success probability is 
${\alpha_{\solving}^2 > {1}/{4}}$.

To realize $F_{\chi}(\theta_{s(n)})$ in a distributed manner,
where $\chi(\vec{x})=1$ if $\vec{x}$ is solving
and $\chi(\vec{x})=0$ otherwise, 
each party wants to
multiply the amplitude of any basis state $\ket{\vec{x}}$ for $\chi(\vec{x})=1$
by a factor of
$e^{i\frac{1}{n}\theta_{s(n)}}$, where $s(n)=(1 - 1/n)^{n-1}$.  
This will multiply the amplitude of the basis state by 
a factor of $e^{i\theta_{s(n)}}$ as a whole.
At this point, however, no party can check 
if $\chi(\vec{x})=1$ for each basis state $\ket{\vec{x}}$, 
since he knows only the content of his $\reg{R}$.
Thus, every party runs 
the $H_1$-algorithm
with $\reg{R}$ and 
$\reg{S}$,
which sets 
the content of $\reg{S}$ to ``$\true$'' 
if the number of $1$'s among the contents of $\reg{R}$'s of all parties
is exactly one and sets it to ``$\false$'' otherwise
(recall that the $H_1$-algorithm computes $H_1$ for each basis state in
 a superposition). This operation transforms the state as follows:
\[
\ket{\Psi}\ket{\mbox{``$\true$''}}^{\otimes n}
\mapsto \alpha_{\solving}\ket{\Psi_{\solving}}\ket{\mbox{``$\true$''}}^{\otimes n}
+ \alpha_{\unsolving}\ket{\Psi_{\unsolving}}\ket{\mbox{``$\false$''}}^{\otimes n},
\]
where the last $n$ qubits are those in $\reg{S}$'s.
Every
party then multiplies the amplitude of each basis state by a factor of
$e^{i\frac{1}{n}\theta_{s(n)}}$, if the content of $\reg{S}$ is ``$\true$''
(here, no party measures $\reg{S}$; every party
just performs the phase-shift operator controlled by the qubit in $\reg{S}$).
Namely, the state over $\reg{R}$'s and $\reg{S}$'s of all parties is transformed into
\[
(e^{i\frac{1}{n}\theta_{s(n)}})^n\alpha_{\solving}\ket{\Psi_{\solving}}\ket{\mbox{``$\true$''}}^{\otimes n}
+ \alpha_{\unsolving}\ket{\Psi_{\unsolving}}\ket{\mbox{``$\false$''}}^{\otimes n}.
\]
Finally,
every party inverts every computation and communication
of the $H_1$-algorithm to
disentangle $\reg{S}$.

The implementation of $F_{0}(\phi_{s(n)})$ is
similar to that of $F_{\chi}(\theta_{s(n)})$,
except that 
$F_{0}(\phi_{s(n)})$ multiplies
the all-zero basis state $\ket{0\cdots 0}$ by $e^{i\phi_{s(n)}}$.
First, every party runs
the $H_0$-algorithm
with $\reg{R}_0$ and $\reg{S}^\prime$,
which 
sets the content of $\reg{S}^\prime$ to ``$\true$'' in the case of the all-zero state,
and sets it to ``$\false$'' otherwise.
Next, every
party multiplies the amplitude of the all-zero state by a factor of
$e^{i\frac{1}{n}\phi_{s(n)}}$, if the content of $\reg{S}^\prime$ is ``$\true$''.  
Finally,
every party inverts every computation and communication
of the $H_0$-algorithm to
disentangle $\reg{S}^\prime$.

More precisely, every party sets his classical variable $\status$ to
$\mbox{``$\eligible$''}$, and runs Algorithm~QLE with $\status$ and $n$,
given in Figure~\ref{fig: Algorithm QLE}. 
After the execution of the algorithm, exactly one party has the value $\mbox{``$\eligible$''}$ in $\status$.
Since all communication  is performed to compute $H_0$ and $H_1$ and their inversions,
the algorithm runs in $2(Q^{\operatorname{rnd}}_G(\calH_0)+Q^{\operatorname{rnd}}_G(\calH_1))$ rounds
with bit complexity $2(Q^{\operatorname{bit}}_G(\calH_0)+Q^{\operatorname{bit}}_G(\calH_1))$
for any graph $G\in \calG_n$, where $\calH_0$ and $\calH_1$ are the $H_0$-algorithm and $H_1$ algorithm, respectively,  that we fixed.
This completes the proof of Theorem~\ref{th:LE}.
 \begin{figure}[t]
 {\small
\begin{algorithm*}{Algorithm~QLE}
\begin{description}
\setlength{\itemsep}{-1mm}
\item[Input:]
classical variable 
$\status:=\mbox{``$\eligible$''}$,
and integer ${n}$
\item[Output:]
classical variable $\status\in \{ \mbox{``$\eligible$''}, \mbox{``$\ineligible$"}\}$\\
\end{description}
\begin{step}
\setlength{\itemsep}{0mm}
\item
  Initialize quantum registers $\reg{R}$, $\reg{S}$, and $\reg{S}^\prime$to $\ket{0}$, $\ket{\mbox{``$\true$''}}$, and $\ket{\mbox{``$\true$''}}$ states, respectively.
\item
  If $\status= \mbox{``$\eligible$''}$,   apply 
  ${
     \bfA
     =
     {
       \scriptsize
       \frac{1}{\sqrt{n}}
       \begin{pmatrix}
	 \sqrt{n-1} & 1\\
	 1 & - \sqrt{n-1}
       \end{pmatrix}
     }
  }$
  to the qubit in $\reg{R}$ to  generate the quantum state
 ${
 \ket{\psi}
 =
 \sqrt{\frac{n-1}{n}}\ket{0} + \sqrt{\frac{1}{n}}\ket{1}
 }$
 in $\reg{R}$.
\\

\item
  Perform the exact amplitude amplification consisting of the following steps:
  \begin{step}
  \item
    To realize $F_{\chi}(\psi_{s(n)})$ for $s(n)=(1 - 1/n)^{n-1}$, perform the following steps:
    \begin{step}
    \item
      Perform an $H_1$-algorithm with $\reg{R}$ and $\reg{S}$, $n$.
    \item
      Multiply the content of $\reg{R}$ by a factor of  $\exp(i\frac{1}{n}\theta_{s(n)})$ if the content of
      $\reg{S}$ is ``$\true$''.
    \item
      Invert every computation and communication of step 3.1.1 to
      disentangle $\reg{S}$.
    \end{step}
  \item
    Invert the computation of Step 2.
  \item
    To realize $F_{0}(\phi_{s(n)})$, perform the following steps:
    \begin{step}
    \item
      Perform an $H_0$-algorithm with $\reg{R}$, $\reg{S}^\prime$ and $n$.
    \item 
Multiply the content of $\reg{R}$ by a factor of  $\exp(i\frac{1}{n}\phi_{s(n)})$  if the content of
      $\reg{S}^\prime$ is ``$\true$''.
    \item
      Invert every computation and communication of Step 3.3.1 to
      disentangle $\reg{S}^\prime$.
   \end{step}
  \item
    Perform the same operation as is performed in Step 2.
\end{step}
\item
  Measure $\reg{R}$ with respect to basis $\set{\ket{0},\ket{1}}$.
  If the result is $1$, then set $\status$ to ``$\eligible$''.
\item
  Output $\status$.
\end{step}
\end{algorithm*}
  \vspace{-1.5\baselineskip}
  \caption{Algorithm QLE}
  \label{fig: Algorithm QLE}
 }
 \end{figure}

%%% Local Variables: 
%%% mode: latex
%%% TeX-master: "KobMatTan10"
%%% End: 

\section{Proof of Theorem~\ref{th:H1}}
\label{subsec:ProofOfH1}
The proof consists of the following two steps:
\begin{itemize}
\item Reduce computing $H_1$ to computing $H_0$ and the consistency function $C_S$, 
where $C_S$ is a Boolean function that is $\true$ if and only if
a subset (specified by $S$) of all parties has the same classical value
(its formal definition will be given later).
\item Reduce computing $C_S$ to computing $H_0$.
\end{itemize}
Actually, the second step is almost trivial. We start with the first step.

\subsection{Basic Idea}
\label{sec:basic-idea}
Suppose that every party $i$ is given a Boolean variable $x_i$.  
We can probabilistically compute $H_1(\vec{x})$ with the following classical algorithm,
where $\vec{x}=(x_1,\dots ,x_n)$:
Every party $i$ with $x_i=1$ sets a variable $r_i$ to $0$ or $1$ each with probability $1/2$ and sends $r_i$ to all parties
(by taking $\delta$ rounds for the diameter $\delta$ of the underlying graph); 
every party $i$ with
$x_i=0$ sets variable $r_i$ to ``$\ast$'' and sends $r_i$ to all parties.  
It is not difficult to see that the following three hold: (i) if
$\abs{\vec{x}} = 0$, every party receives only ``$\ast$'', (ii) if
$\abs{\vec{x}} = 1$, either no party receives ``$1$'' or no party
receives ``$0$'', and (iii) if $\abs{\vec{x}} = t \geq 2$, every party
receives both ``$0$'' and ``$1$'' with probability $1 - 2/2^t$.
Therefore, every party can
conclude that $H_1(\vec{x})=\true$ 
($H_1(\vec{x})=\false$)
with probability one if
$\abs{\vec{x}}=1$ ($\abs{\vec{x}}=0$) and that $H_1(\vec{x})=\false$ with probability
$1-2/2^t\ge 1/2$ if $\abs{\vec{x}}=t \geq 2$. 
Roughly speaking, 
our quantum algorithm for computing $H_1$ is obtained by first quantizing this probabilistic
algorithm and then applying the exact quantum amplitude amplification to boost
the success probability to one. 
More concretely, we amplify the probability $p$ that there are both $0$ and $1$ among all $r_i$'s by using the exact amplitude amplification.
Let $p_{\text{init}}$ and $p_{\text{final}}$ be the values of $p$ before and after, respectively, applying the amplitude amplification. Obviously, if  $p_{\text{init}}=0$,
then $p_{\text{final}}=0$ also.
Hence, for $\abs{\vec{x}}\le 1$, $p_{\text{final}}=0$.
For $\abs{\vec{x}}\geq 2$, $p$ could be boosted to one if the exact value of 
$p_{\text{init}}$  were known to every party. However, $p_{\text{init}}$ is determined by $t$, the value of which may be harder to compute than to just decide whether $t=1$ or not.
Therefore, instead of actual $t$,
we run the amplitude amplification for each $t^{^\prime}:=2,\dots, n$, a \emph{guess} of $t$,  in parallel. We can then observe that exactly one of the $(n-1)$ runs boosts $p$ to one if and only if $\abs{\vec{x}}\geq 2$.

\subsection{Terminology}
Suppose that each party $i$ has a bit  $x_i$,
i.e., the $n$ parties share $n$-bit string ${\vec{x} = (x_1,x_2, \dots ,x_n)}$.
For convenience, we may consider that each $x_i$ expresses an integer,
and identify string $x_i$ with the integer it expresses.
For an index set ${S \subseteq \{1, \ldots, n\}}$,
string $\vec{x}$ is said to be \emph{consistent} over $S$
if $x_i$ is equal to $x_j$
for all $i,j$ in $S$.
Otherwise $\vec{x}$ is said to be \emph{inconsistent} over $S$.
Here, index set $S$ is used just for the definition (recall that no party has an index or identifier in the anonymous setting).
Formally, we assume that every party has a variable 
$z \in \set{\mbox{``$\marked$''},\mbox{``$\unmarked$''}}$,
and $S$ is defined as the set of all parties with 
$z=\mbox{``$\marked$''}$.
If $S$ is the empty set, any $\vec{x}$ is said to be consistent over $S$.
We also say that an $n$-qubit pure state
${\ket{\psi} = \sum_{\vec{x}\in \set{0,1}^{n}} \alpha_{\vec{x}} \ket{\vec{x}}}=
\sum_{\vec{x}\in \set{0,1}^{n}} \alpha_{\vec{x}} \ket{x_1}\otimes \dots \otimes \ket{x_n}
$
shared by the $n$ parties is \emph{consistent (inconsistent)} over $S$
if ${\alpha_{\vec{x}} \neq 0}$ only for $\vec{x}$\! 's that are consistent (inconsistent) over $S$
(there are pure states that are neither consistent nor inconsistent over $S$, but we do not need to define such states).

We next define the consistency function $C_S\colon\set{0,1}^{n}\to \set{\mbox{``$\consistent$''}, \mbox{``$\inconsistent$''}}$,
which decides if a given string $\vec{x}\in \set{0, 1}^{n}$ distributed over $n$ parties
is consistent over $S$. Namely, $C_S(\vec{x})$ returns ``$\consistent$'' if $\vec{x}$ is consistent over $S$ 
and ``$\inconsistent$'' otherwise.

\subsection{The $\boldsymbol{H_1}$-Algorithm}
As in the previous section, we fix an $H_0$-algorithm and a
$C_S$-algorithm, which we
use in a black-box manner.
At the beginning of the algorithm, every party prepares two one-qubit registers $\reg{X}$ and $\reg{Y}$.
We shall describe an $H_1$-algorithm that exactly computes function $H_1$ over
the contents of $\reg{X}$'s and sets
the content of each $\reg{Y}$ to the function value.
Here, we assume that registers $\reg{Y}$'s are initialized to $\ket{\mbox{``$\true$''}}$
for an orthonormal basis 
$\set{\ket{\mbox{``$\true$''}},\ket{\mbox{``$\false$''}} }$
of $\Complex ^2$.
We basically follow the idea in Section~\ref{sec:basic-idea}
to reduce computing $H_1$ to computing the binary-valued functions $H_0$ and $C_S$.
However, the idea actually represents a three-valued function, i.e., distinguishes among three cases: 
$\abs{\vec{x}}=0$, $\abs{\vec{x}}=1$, and $\abs{\vec{x}}\geq 2$.
Thus, we cast the idea into two yes-no tests. Namely, the algorithm first tests if $\abs{\vec{x}}$ is 0 or not. If
$\abs{\vec{x}}=0$, then it concludes
$H_1(\vec{x})=\mbox{``$\false$''}$. The algorithm then performs
another test to decide if $\abs{\vec{x}}\leq 1$ or $\abs{\vec{x}}\geq 2$, which determines $H_1(\vec{x})$.

\subsubsection{First Test} 
To test if $\abs{\vec{x}} = 0$, each party prepares a single-qubit
register $\sfS_0$, the content of which is initialized to
$\ket{``\true''}$. Each party then performs the $H_0$-algorithm to
exactly compute the value of $H_0$ over the contents of $\sfX$'s, and
stores the computed value in each $\sfS_0$.

From the definition of the $H_0$-algorithm, this transforms the state in $\sfX$'s and $\sfS_0$'s as follows:
\begin{align*}
\bigotimes_{i=1}^n
\big(
\ket{x_i}_{\reg{X}}
\ket{\mbox{``$\true$''}}_{\reg{Y}}
\ket{\mbox{``$\true$''}}_{\reg{S}_0}
\big)
&\mapsto
\bigotimes_{i=1}^n
\big(
\ket{x_i}_{\reg{X}}
\ket{\mbox{``$\true$''}}_{\reg{Y}}
\ket{H_0(\vec{x})}_{\reg{S}_0}
\big),\\
\intertext{by rearranging registers,}
&=
\underbrace{\ket{\vec{x}}}_{\reg{X}\text{'s}}
\ket{\mbox{``$\true$''}}_{\reg{Y}}^{\otimes n}
\ket{H_0(\vec{x})}_{\reg{S}_0}^{\otimes n}.
\end{align*}
If the content of $\reg{S}_0$ is ``$\true$'', then the content of $\reg{Y}$ will be set to $\mbox{``$\false$''}$ later
(because this means $\abs{\vec{x}}=0$).

\subsubsection{Second Test} 
Next each party tests if $\abs{\vec{x}}\le 1$ or $\abs{\vec{x}}\geq  2$ with certainty.
Recall the probabilistic algorithm in which
every party $i$ sets a variable $r_i$ to $0$ or
$1$ each with probability $1/2$ if $x_i=1$ and
sets variable $r_i$ to ``$\ast$'' if $x_i=0$,
and then sends $r_i$ to all parties.
Our goal is to amplify the probability $p$ that 
that there are both $0$ and $1$ among all $r_i$'s by using the exact amplitude amplification.
The difficulty is that 
no party knows the value of $p_{\operatorname{init}}$ ($=1-2/2^{\abs{\vec{x}}}$).
The test thus uses a \emph{guess} $t$ of $\abs{\vec{x}}$ and tries to 
amplify $p$
assuming that $p_{\operatorname{init}}=1-2/2^{t}$.
If $t=\abs{\vec{x}}$,
then the procedure obviously outputs the correct answer with probability one. 
If $t\neq \abs{\vec{x}}$,
the procedure may output the wrong answer.
As will be proved later, however, 
we can decide if $\abs{\vec{x}}\le 1$ or $\abs{\vec{x}}\geq 2$ without error
from the outputs of $(n-1)$-runs of the test for $t=2,\dots, n$,
which are performed in parallel.

We now describe the test procedure for each $t$. 
Assume that one-qubit register $\reg{Z}_t$
is initialized to $\ket{\mbox{``$\unmarked$ ''}}$.
The initial state is thus
\[
\sum _{\vec{x}\in\set{0,1}^n} \alpha _{\vec{x}}
\bigotimes_{i=1}^n
\big(
\ket{x_i}_{\reg{X}}
\ket{\mbox{``$\unmarked$ ''}}_{\reg{Z}_t}
\big),
\]
where registers $\reg{Y}$ and $\reg{S}_0$ are omitted to avoid complication.

The base algorithm $\calA$ (to be amplified) is described as follows.
If the content of $\reg{X}$ is $1$, 
the party flips the content of $\reg{Z}_t$ to ``$\marked$'',
where $\set{\ket{\mbox{``$\marked$ ''}},\ket{\mbox{``$\unmarked$ ''}}}$
is an orthonormal basis in $\Complex ^2$.
This operation just copies the contents of $\reg{X}$ to those of 
$\reg{Z}_t$ 
(in the different orthonormal basis)
for parallel use over all $t$.
The state is thus, for any fixed $\vec{x}$,
\[
\bigotimes_{i=1}^n
\big(
\ket{x_i}_{\reg{X}}
\ket{z_t(x_i)}_{\reg{Z}_t}
\big)
=
\Big(
\ket{1}_{\reg{X}}
\ket{\mbox{``$\marked$ ''}}_{\reg{Z}_t}
\Big)^{\otimes \abs{S}}
\otimes
\Big(
\ket{0}_{\reg{X}}
\ket{\mbox{``$\unmarked$ ''}}_{\reg{Z}_t}
\Big)^{\otimes (n-\abs{S})},
\]
 where $z_t(x_i)\in \set{\ket{\mbox{``$\marked$
       ''}},\ket{\mbox{``$\unmarked$ ''}}}$ is the content of
 $\reg{Z}_t$ when the content of $\reg{X}$ is $x_i$, and 
$S$ is
the set of the parties whose $\reg{Z}_t$ is in the state
$\ket{\mbox{``$\marked$ ''}}$ (note that $\abs{S}=\abs{\vec{x}}$).

If the content of $\reg{Z}_t$ is``$\marked$'',
apply the Hadamard operator 
$\bfH
=
\frac{1}{\sqrt{2}}
\left(
\begin{smallmatrix}
  1 & 1\\
1 & -1
\end{smallmatrix}
\right)
$ 
to the qubit in $\reg{R}_t$ to create $(\ket{0}+\ket{1})/\sqrt{2}$
(note that register $\reg{R}_t$ of each party $i$ is the quantum equivalent of $r_i$\footnote{
Here, the contents of $\reg{R}_t$'s of ``unmarked'' parties are set to $\ket{0}$, while the classical equivalents, variables $r_i$'s, of the parties are set to ``$\ast$'' (instead of 0). Actually,  the symbol ``$\ast$'' is used to distinguish 
between $\abs{\vec{x}}=0$ and $\abs{\vec{x}}=1$. However, we do not need it any longer due to the first test.}).
The state is now represented as, for the $\vec{x}$,
\[
\Big(
\ket{1}_{\reg{X}}
\ket{\mbox{``$\marked$ ''}}_{\reg{Z}_t}
\frac{\ket{0}_{\reg{R}_t}+\ket{1}_{\reg{R}_t}}{\sqrt{2}}
\Big)^{\otimes \abs{S}}
\otimes
\Big(
\ket{0}_{\reg{X}}
\ket{\mbox{``$\unmarked$ ''}}_{\reg{Z}_t}
\ket{0}_{\reg{R}_t}
\Big)^{\otimes (n-\abs{S})}.
\]
By rearranging registers, we have
\[
\underbrace{\ket{\vec{x}}}_{\reg{X}'s}
\underbrace{\ket{z_t(\vec{x})}}_{\reg{Z}_t's}
\left(
\frac{\ket{0}_{\reg{R}_t}+\ket{1}_{\reg{R}_t}}{\sqrt{2}}
\right)^{\otimes \abs{S}}
\ket{0}_{\reg{R}_t}^{\otimes (n-\abs{S})}
= 
\underbrace{\ket{\vec{x}}}_{\reg{X}\mbox{'s}}
\underbrace{\ket{z_t(\vec{x})}}_{\reg{Z}_t\mbox{'s}}
\underbrace{\ket{\psi_t(\vec{x})}}_{\reg{R}_t\mbox{'s}},
\]
where $\ket{z_t(\vec{x})}$ is the $n$-tensor product of 
$\ket{\mbox{``$\marked$ ''}}$ or $\ket{\mbox{``$\unmarked$ ''}}$ corresponding to $\vec{x}$,
and 
\[
\ket{\psi_t(\vec{x})}
=
\left(
\frac{1}{\sqrt{2^{\abs{S}}}}
\sum_{\vec{y}\in \set{0,1}^{\abs{S}}}
\ket{\vec{y}}
\right)
\ket{0}^{\otimes (n-\abs{S})}.
\]
This is the end of the base algorithm $\calA$.

We then boost the amplitudes of the basis states superposed in 
$\ket{\psi_t(\vec{x})}$ such that
there are both $\ket{0}$ and $\ket{1}$ in $\reg{R}_t$'s of parties in $S$,
i.e., the amplitudes of the states that are inconsistent over $S$,
with amplitude amplification.
Here, function $\chi$ in Theorem~\ref{th: quantum amplitude amplification} is the consistency function $C_S$
and $a(t)=1-2\left(\frac{1}{2}\right)^{t}$ is used as the success probability $a$.
For convenience, we express 
$\ket{\psi_t(\vec{x})}$ as
\[
\ket{\psi_t(\vec{x})}
=
\ket{\psi _{\inconsistent}} 
+
 \ket{\psi _{\consistent}},
\]
where,
\begin{align*}
\ket{\psi _{\inconsistent}}
&=
\left(\frac{1}{\sqrt{2^{\abs{S}}}}\sum _{\vec{y}\in \set{0,1}^{|S|}: \abs{\vec{y}}\neq 0,|S|} \ket{\vec{y}}\right)\ket{0}^{\otimes(n-|S|)},\\
\ket{\psi _{\consistent}}
&=
\frac{1}{\sqrt{2^{\abs{S}}}}
\left(
\ket{0}^{\otimes \abs{S}}+
\ket{1}^{\otimes\abs{S}}
\right)
\ket{0}^{\otimes(n-|S|)}.
\end{align*}

To realize $F_{\chi}(\theta_{a(t)})$, 
every party prepares a single-qubit register $\reg{S}_t$ initialized to $\ket{\mbox{``$\consistent$''}}$ and
        then performs the next operations:
(1) Perform a $C_S$-algorithm with $\reg{R} _t $, $\reg{S}_t$ and $\reg{Z}_t$,
which computes $C_S$ for each basis state $\ket{\vec{y}}\ket{0}^{n-\abs{S}}$ of
$\ket{\psi_t(\vec{x})}$
and sets the content of $\reg{S}_t$ to value of $C_S$;
(2) Multiply the amplitude of each basis state of $\reg{R}_t$ by a factor of $\exp\big(i\frac{\theta_{a(t)}}{n}\big)$
if the content of $\reg{S}_t$ is ``$\inconsistent$'';
(3) Finally invert every computation and communication of (1) to disentangle $\reg{S}_t$.
The state evolves with the above operations as follows:
\begin{align*}
\lefteqn{\ket{z_t(\vec{x})}\ket{\psi_t(\vec{x})}\ket{\mbox{``$\consistent$''}}^{\otimes n}}\\
&\mapsto
\ket{z_t(\vec{x})}
\left(
\ket{\psi _{\inconsistent}} \ket{\mbox{``$\inconsistent$''}}^{\otimes n}
+
\ket{\psi _{\consistent}} \ket{\mbox{``$\consistent$''}}^{\otimes n}
\right)
\\
&\mapsto
\ket{z_t(\vec{x})}
\left(
\big(e^{i\frac{\theta_{a(t)}}{n}}\big)^n\ket{\psi _{\inconsistent}} \ket{\mbox{``$\inconsistent$''}}^{\otimes n}
+
\ket{\psi _{\consistent}} \ket{\mbox{``$\consistent$''}}^{\otimes n}
\right)\\
&\mapsto
\ket{z_t(\vec{x})}
\left(
\left(e^{i\theta_{a(t)}}
\ket{\psi _{\inconsistent}}
+
\ket{\psi _{\consistent}}
\right) \ket{\mbox{``$\consistent$''}}^{\otimes n}
 \right).
\end{align*}
We have now finished the first operation, $F_{\chi}(\theta_{a(t)})$, of 
$-\calA F _{0}(\phi_{a(t)})\calA^{-1}F _{\chi}(\theta_{a(t)})$.

Then $\calA^{-1}$ is performed.
Operation $F_{0}(\phi_{a(t)})$ can be realized 
with the $H_0$-algorithm
in the same way as in Algorithm QLE in the previous section.
Finally, perform operation $\calA$ again. This is the end of the amplitude amplification.
In summary, the state over
$\reg{Z_t}$'s and $\reg{R}_t$'s is transformed
as follows:
\[
\ket{z_t(\vec{x})}\ket{\psi_t(\vec{x})}
\mapsto 
\ket{z_t(\vec{x})}\ket{\psi^{\prime}_t(\vec{x})},
\]
where
$\ket{\psi^{\prime}_t(\vec{x})}$ is expressed as in the following claim.
\begin{claim}
\label{claim:1}
\[
\ket{\psi^{\prime}_t(\vec{x})}=
\begin{cases}
  \sqrt{\frac{2^{\abs{S}}}{2^{\abs{S}-2}}} \ \ket{\psi _{\inconsistent}} & (\text{$\abs{S}\ge 2$ and $t=\abs{S}$}),\\
\beta \ket{\psi _{\inconsistent}} + \gamma\ket{\psi _{\consistent}}
& (\text{$\abs{S}\ge 2$ and $t\neq \abs{S}$}),\\ 
\ket{\psi _{\consistent}} & (\abs{S}\le 1 \text{ and all } $t$),
\end{cases}
\]
for some $\beta, \gamma\in \Complex$.
\end{claim}
\begin{proofof}{Claim~\ref{claim:1}}
If $\abs{S}\ge 2$ and $t=|S|$,
the claim follows from Theorem~\ref{th: quantum amplitude amplification}.
If $\abs{S}\ge 2$ and $t\neq |S|$, the claim is trivial.
If $\abs{S}\le 1$, 
then $\ket{\psi_t(\vec{x})}= \ket{\psi _{\consistent}}$; thus, $\ket{\psi^{\prime}_t(\vec{x})}=\ket{\psi _{\consistent}}$.
\end{proofof}

Each party then prepares a new quantum register $\reg{S}_t^{\prime\prime}$ (initialized to $\ket{\mbox{``$\consistent$''}}$) and
performs again
the $C_S$-algorithm with $\reg{R} _t $, $\reg{S}_t^{\prime\prime}$ and $\reg{Z}_t$,
which transforms the state as follows:
\[
\underbrace{\ket{z_t(\vec{x})}}_{\reg{Z}_t\text{'s}}
\underbrace{\ket{\psi^{\prime}_t(\vec{x})}}_{\reg{R}_t\text{'s}}
\underbrace{\ket{\mbox{``$\consistent$''}}^{\otimes n}}_{\reg{S}_t^{\prime\prime}\text{'s}}
\to 
\ket{z_t(\vec{x})}
\underbrace{\ket{\Psi_t(\vec{x})}}_{\reg{R}_t\text{'s, }\reg{S}_t^{\prime\prime}\text{'s}},
\]
where 
\[
\ket{\Psi_t(\vec{x})}=
\begin{cases}
  \sqrt{\frac{2^{\abs{S}}}{2^{\abs{S}-2}}} \ \ket{\psi _{\inconsistent}}\ket{\mbox{``$\inconsistent$''}}^{\otimes n} & (\text{$\abs{S}\ge 2$ and $t=\abs{S}$}),\\
\beta \ket{\psi _{\inconsistent}}\ket{\mbox{``$\inconsistent$''}}^{\otimes n}  + \gamma\ket{\psi _{\consistent}}\ket{\mbox{``$\consistent$''}}^{\otimes n} 
& (\text{$\abs{S}\ge 2$ and $t\neq \abs{S}$}),\\ 
\ket{\psi _{\consistent}}\ket{\mbox{``$\consistent$''}}^{\otimes n}  & (\abs{S}\le 1 \text{ and all } t).
\end{cases}
\]

\subsubsection{Final Evaluation }
After the first test and the second tests for $t=2,\dots ,n$, the state is now
\[
\underbrace{\ket{\vec{x}}}_{\reg{X}}
\otimes
\underbrace{\ket{\true}^{\otimes n}}_{\reg{Y}\text{'s}}
\otimes
\bigg(
\underbrace{\ket{H_0(\vec{x})}^{\otimes n} }_{\reg{S}_0\text{'s}}
\bigg)
\otimes
\bigg(
\bigotimes _{t=2} ^n 
\ket{z_t(\vec{x})}
\underbrace{\ket{\Psi_t(\vec{x})}}_{\reg{R}_t\text{'s, }\reg{S}_t^{\prime\prime}\text{'s}}
\bigg).
\]

Recall that every party has registers $\reg{Y}$, $\reg{S}_0$, $\reg{S}_t^{\prime\prime}$ for $t=2,\dots, n$.
In the final step of our algorithm for computing $H_1$, every party concludes the value of $H_1(\vec{x})$
from the contents of $\reg{S}_0$ and $\reg{S}_t^{\prime\prime}$'s as follows:
\begin{itemize}
\item If either the content of $\reg{S}_0$ is ``$\true$'' or the
    content of $\reg{S}_t^{\prime\prime}$ is ``$\inconsistent$'' for some
    $t\in \set{2,\dots, n}$, then every party sets the content of $\reg{Y}$ to ``$\false$''.
\end{itemize}
It is not difficult to show the correctness.
If the content of $\reg{S}_0$ is ``$\true$'', then the value of $H_1(\vec{x})$ is obviously ``$\false$''
(because $\abs{\vec{x}}=0$). 
Suppose that the content of $\reg{S}_0$ is ``$\false$'', i.e., $\abs{S}\neq 0$.
From the definition of $\ket{\Psi_t(\vec{x})}$, we can observe the following facts:
(1) If $\abs{\vec{x}}:=\abs{S}=1$, then 
the contents of $\reg{S}_t^{\prime\prime}$ are ``$\consistent$'' for all $t=2,\dots ,n$.
(2) If $\abs{\vec{x}}:=\abs{S}\ge 2$,
then the content of $\reg{S}_t^{\prime\prime}$ is ``$\inconsistent$'' for some $t\in \set{2,\dots ,n}$.
More precise description of our algorithm is given 
in Figure~\ref{fig:H_1-Algorithm}.

 \begin{figure}[htb]
   \centering
{\small
\begin{algorithm*}{$\boldsymbol{H_1}$-Algorithm}
  \begin{description}
\setlength{\itemsep}{0pt}
  \item[Input:] Single-qubit registers $\reg{X}$ and $\reg{Y}$ (W.L.O.G., initialized to $\ket{\mbox{``$\true$''}}$), an integer $n$.
  \item[Output:] Single-qubit registers $\reg{X}$ and $\reg{Y}$.
  \end{description}
  \begin{step}
\setlength{\itemsep}{0pt}
  \item Initialize $n$ single-qubit registers $\reg{R}_0,\reg{R}_2,\dots ,\reg{R}_n$ to $\ket{0}$.
  \item Prepare a single-qubit register $\reg{S}_0$ and then perform the following steps of the first test:
    \begin{step}
    \item 
      Copy the content of $\reg{X}$ to that of $\reg{R}_0$
      in the $\set{\ket{0},\ket{1}}$ basis (i.e., apply CNOT to $\reg{R}_0$ 
with $\reg{X}$  as control).
    \item 
      Perform an $H_0$-algorithm with $\reg{R}_0$, $\reg{S}_0$ and $n$, which computes $H_0$ over the contents of $\reg{R}_0$'s of all parties  and store the result into $\reg{S}_0$.\\
    \end{step}
  \item   Perform the following steps of the second test for $t=2,\dots , n$ in parallel:
    \begin{step}
    \item If the content of $\reg{X}$ is $1$, set the content of $\reg{Z}_t$ to ``$\marked$ ''; otherwise
set it to ``$\unmarked$ ''.
    \item 
      If the content of $\reg{Z}_t$ is ``$\marked$ '', apply the Hadamard operator on the qubit in $\reg{R}_t$ (to create $\frac{\ket{0}+\ket{1}}{\sqrt{2}}$).
      \item
	To realize $F_{\chi}(\psi_{a(t)})$, 
        prepare a single-qubit quantum register $\reg{S}_t$ and perform the following operations:
      \begin{step}
    \item 
     Perform a $C_S$-algorithm with $\reg{R} _t $, $\reg{S}_t$, $\reg{Z}_t$ and $n$, which computes $C_S$ over the contents of $\reg{R}_t$'s of all parties for $S$ defined by the contents of $\reg{Z}_t$'s,
and stores the result into $\reg{S}_t$.
   \item Multiply  the state of $\reg{R}_t$ by a factor of $e^{i\frac{1}{n}\psi_{a(t)}}$ if the content of $\reg{S}_t$ is ``$\inconsistent$'',
        where 
        $a(t)$ is the probability of measuring inconsistent states in $\left(
          \frac{\ket{0}+\ket{1}}{\sqrt{2}} \right)^{\otimes t}$, i.e.,
        $1-2\left(\frac{1}{2}\right)^{t}$.
        \item
	  Invert every computation and communication of Step 3.3.1 to disentangle $\reg{S}_t$.
      \end{step}
    \item Invert the computation in Step 3.2.
      \item
	To realize $F_{0}(\phi_{a(t)})$, 
        prepare a single-qubit quantum register $\reg{S}^{\prime}_t$ and
        perform the following  operations:
        \begin{step}
    \item Perform the $H_0$-algorithm with $\reg{R} _t$, $\reg{S}^{\prime}_t$ and $n$,
which computes $H_0$ over the contents of $\reg{R}_0$'s of all parties and store the result into $\reg{S}^{\prime}_t$.\\
    \item Multiply the state of $\reg{R}_t$ by a factor of $e^{i\frac{1}{n}\phi_{a(t)}}$ if the content of 
        $\reg{S}_t^\prime$ is ``$\true$''.
        \item
	  Invert every computation and communication of Step 3.5.1 to disentangle $\reg{S}^\prime_t$.
        \end{step}
        \item Perform the same operation as in Step~3.2
        \item Prepare a fresh single-qubit register $\reg{S}_t^{\prime\prime}$, and
     perform a $C_S$-algorithm with $\reg{R} _t $, $\reg{S}_t^{\prime\prime}$, $\reg{Z}_t$ and $n$, which computes $C_S$ over the contents of $\reg{R}_t$'s of all parties for $S$ defined by the contents of $\reg{Z}_t$'s,
and stores the result into $\reg{S}_t^{\prime\prime}$.
    \end{step}
  \item If either the content of $\reg{S}_0$ is ``$\true$'' or the
    content of $\reg{S}_t^{\prime\prime}$ is ``$\inconsistent$'' for some
    $t\in \set{2,\dots, n}$, then turn $\reg{Y}$
    over (i.e., transform the state $\ket{\mbox{``$\true$''}}$ of $\reg{Y}$ into $\ket{\mbox{``$\false$''}}$).
  \item Invert every computation and communication of Steps 2 and 3 to disentangle all registers except $\reg{X}$ and $\reg{Y}$.
  \item Output $\reg{X}$ and $\reg{Y}$, and then halt.
  \end{step}
\end{algorithm*}
}
 \vspace{-3mm}
   \caption{$H_1$-Algorithm}
   \label{fig:H_1-Algorithm}
 \end{figure}

\begin{lemma}
\label{lm:H_1-Algorithm-Reduction}
For any graph $G\in \calG_n$,
if every party knows the number $n$ of parties.
there is an $H_1$-algorithm that runs in
$O(Q^{\operatorname{rnd}}_G(\calH_0)+Q^{\operatorname{rnd}}_G(\calC_S))$ rounds
with bit complexity $O(n(Q^{\operatorname{bit}}_G(\calH_0)+Q^{\operatorname{bit}}_G(\calC_S)))$,
where $\calH_0$ and $\calC_S$ are any $H_0$-algorithm and any $C_S$-algorithm, respectively.
\end{lemma}
\begin{proof}
The correctness follows from the above description of the algorithm. For the complexity, 
all communications are performed for computing $H_0$ 
and then computing $C_S$ for $t=2,...,n$ in parallel. Therefore, the lemma follows.
\end{proof}

\subsection{Computing $\boldsymbol{C_S}$ with Any $\boldsymbol{H_0}$-Algorithm}
We now show that computing $C_S$ is reducible to computing $H_0$.
\begin{lemma}
\label{lm:reduction-consisency_to_H0}
For any graph $G\in \calG_n$,
there is a $C_S$-algorithm  that runs in $O(Q_G^{\operatorname{M}}(\calH_0))$ rounds
with bit complexity $O(Q_G^{\operatorname{M}}(\calH_0))$, where $\calH_0$ is any $H_0$-algorithm.
\end{lemma}
\begin{proof}
Function $C_S$ can be computed by first 
computing in parallel $H_0$ and $H_{\abs{S}}$ over the input bits
of the parties associated with $S$, and then computing OR of them. 
To compute 
$H_0$ over the $\abs{S}$ bits with any $H_0$-algorithm over $n$ bits,
every party $i$ with $i\not\in S$ sets his input to $0$,
and all parties then run the $H_0$-algorithm.
Similarly, (the negation of)
$H_{\abs{S}}$ over the $\abs{S}$ bits can be computed except that
every party $i$ with $i\in S$ negates his/her input.
\end{proof}

Lemmas~\ref{lm:H_1-Algorithm-Reduction} and \ref{lm:reduction-consisency_to_H0}
imply that,
for any graph $G\in \calG_n$,
there is an $H_1$-algorithm that runs in
$O(Q^{\operatorname{rnd}}_G(H_0))$ rounds
with bit complexity  $O(n\cdot Q^{\operatorname{bit}}_G(H_0))$.
This completes the proof of Theorem~\ref{th:H1}.

Theorem~\ref{th:H1} can easily be generalized to the case where
only an upper bound $N$ of $n$ is given to every party.
Suppose that we are given an $H_0$-algorithm that works for a given
upper bound $N$ of $n$.
The proof of Lemma~\ref{lm:reduction-consisency_to_H0} then implies 
that there exists a $C_S$-algorithm
that can work even if only an upper bound $N$ is given.
We can thus make 
an $H_1$-algorithm 
that works for the upper bound $N$,
by performing the first test and then the second tests for $t=2,\dots , N$ in parallel.

\begin{theorem}
\label{th:H1_Upperbound}
If only an upper bound $N$ of the number $n$ of parties is provided to each party,
function $H_1$ can exactly be computed
without intermediate measurements
for any possible quantum state as input
in 
$O(Q^{\operatorname{rnd}}(H_0))$ rounds
with bit complexity 
$O(N\cdot Q^{\operatorname{bit}}(H_0))$
on  an anonymous 
quantum network of any unknown topology.
\end{theorem}

%%% Local Variables: 
%%% mode: latex
%%% TeX-master: "KobMatTan10"
%%% End: 

\section{Improved Algorithm for $\boldsymbol{\LE_n}$}
As an application of Theorems~\ref{th:LE} and \ref{th:H1}, we present a quantum algorithm that exactly solves
$\LE_n$, which runs with less round complexity
than the existing algorithms
while keeping the best bit complexity.

\begin{proofof}{Corollary~\ref{cr:LE}}  
We first give a simple $H_0$-algorithm in order to apply Theorems~\ref{th:LE} and \ref{th:H1}. 
The algorithm is a straight-forward quantization of the following deterministic algorithm:
Every party sends his input bit to each adjacent party 
(and keep the information of the bit for himself).
Every party then computes the OR of all the bits he received and the bit kept by himself
and sends the resulting bit to each adjacent party
(and keep the information of the bit for himself).
By repeating this procedure $\Delta$ times for an upper bound $\Delta$ of the network diameter,
every party can know the OR of all bits and thus the value of $H_0$.
This classical algorithm can easily be converted to the quantum equivalent
with the same complexity (up to a constant factor).

Thus, we have proved the following claim.
\begin{claim}
\label{cl:H_0-Algorithm-complexity}
Let $G$ be any graph in $\calG_n$, and let $m$ be the number of edges in $G$.
Then,
there is an $H_0$-algorithm that runs
in $O(\Delta)$ rounds with bit complexity
$O(\Delta m)$ on an anonymous quantum network of the underlying graph $G$
 (i.e., $Q^{\operatorname{rnd}}_G(\calH_0)= O(\Delta)$ and  $Q^{\operatorname{bit}}_G(\calH_0)= O(\Delta m)$
for some $H_0$-algorithm $\calH_0$)
if the upper bound $\Delta$ of the diameter of $G$ is given to each party.
\end{claim}
Corollary~\ref{cr:LE} follows from Theorems~\ref{th:LE}, \ref{th:H1} and Claim~\ref{cl:H_0-Algorithm-complexity}
with the trivial upper bound $n$ of $\Delta$.
\end{proofof}

Corollary~\ref{cr:LE} improves the complexity of the existing quantum algorithms for $\LE_n$ in Ref.~\cite{TanKobMat05STACS}.
For particular classes of graphs, it is known that $H_1$ can be computed as efficiently as $H_0$.
In this case, a direct application of Theorem~\ref{th:LE} gives a better
bound. For a ring network, both $H_0$ and $H_1$ can be computed in $O(n)$ rounds with bit complexity $O(n^2)$.

More generally, 
Kranakis et al.~\cite{KraKriBer94InfoComp} developed a random-walk-based classical algorithm
that efficiently computes any symmetric function 
if the stochastic matrix $P$ of the random walk on the underlying graph augmented with self-loops
has a large second eigenvalue (in the absolute sense).
By using this algorithm to compute $H_0$ and $H_1$,
Theorem~\ref{th:LE} yields an efficient algorithm
for the graphs with a large eigenvalue gap.
\begin{corollary}
  Let $G\in \calG_n$ and let $G^\prime$ be the
  graph $G$ with self-loops added to each node. Let $\lambda$ be the
  second largest eigenvalue (in absolute value) of the stochastic
  matrix $P$ associated with $G^\prime$. There is an algorithm that exactly elects a unique leader 
  in $O\left(-\frac{\log n}{\log \lambda}\right)$ rounds with 
 bit complexity $O\left(-\frac{m}{\log \lambda}(\log n)^2\right)$
on an anonymous quantum network with the underlying graph $G$,
where $m$ is the number of edges of $G$.
\end{corollary}
In particular, a unique leader can exactly be elected
in $O(n^{2/d}\log n)$ rounds
with bit complexity $O(n^{1+2/d}\log n)$
for an anonymous quantum $d$-dimensional torus for any integer constant $d\ge 2$, since $-1/\log \lambda\in O(n^{2/d})$.

We next consider a more general setting, in which only an upper bound $N$ of $n$ is given to each party.
In this case, our algorithm can be modified so that it attains the linear round complexity in $N$.
The algorithm, however, has a larger bit complexity than 
than $O(mN^2)$, which is attainable by an existing algorithm.
\begin{corollary}
\label{cr:LE with diameter/upperbound}
Let $G$ be any graph in $\calG_n$, and let $m$ be the number of edges in $G$.
If only an upper bound $N$ of the number $n$ of parties is given to every party,
the leader election problem can exactly be solved 
in $O(N)$ rounds with bit
complexity $O(mN^3)$ on an anonymous quantum network with the underlying graph $G$.
\end{corollary}
\begin{proof}
Theorem~\ref{th:H1_Upperbound} and Claim~\ref{cl:H_0-Algorithm-complexity} imply
that there exist an $H_0$-algorithm and an $H_1$-algorithm 
that work even if only an upper bound $N$ of $n$ is given to each party.

Since Theorem~\ref{th:LE} depends on the high success probability of the
base randomized algorithm (i.e., the algorithm in which every party
flips a coin that gives heads with probability $1/n$), the reduction
works only if $N=n$.
We thus modify the reduction in Theorem~\ref{th:LE} as follows:
(1) We attempt the quantum reduction in 
Theorem~\ref{th:LE} for every guess $n^{\prime}$ of $n$ in parallel,
where $n^{\prime}=2,\dots , N$.
(2) Each attempt is followed by performing the $H_1$-algorithm
to verify that a unique leader is elected.
Observe that for at least one  of 
$n^{\prime}=2,\dots , N$, a unique leader is elected,
which is correctly verified by Step (2) due to Theorem~\ref{th:H1_Upperbound}.
Therefore, the round complexity is $O(N)$ and the bit complexity is $O(mN^3)$.
\end{proof}

%%% Local Variables: 
%%% mode: latex
%%% TeX-master: "KobMatTan10"
%%% End: 

\section{Computing Boolean Functions}
Once a unique leader is elected, a spanning tree can be constructed 
by starting at the leader and traversing the underlying graph 
(e.g., in a depth first manner)
and the leader can assign
a unique identifier to every party by traversing the tree.
Moreover, if a unique leader exists, the underlying graph is recognizable, i.e.,
every party can know the adjacency matrix of the graph, as 
shown in Lemma~\ref{lm:graph-recognition}.
Hence,
it is possible to compute a wider class of Boolean functions than symmetric functions,
i.e.,
all Boolean functions that may depend on the graph topology (but are independent of
the way of assigning unique identifiers to parties). We call such functions \emph{computable functions}.

\begin{lemma}
\label{lm:graph-recognition}
Once a unique leader is elected  on an anonymous quantum network of any topology,
the underlying graph
can be recognized in $O(n)$ rounds with $O(n^3)$ bit complexity.
\end{lemma}
\begin{proof}
    Once a unique leader has been elected, the following procedure can
  recognize the underlying graph.  First construct a spanning tree in
  $O(n)$ rounds with $O(m)$ bit complexity by traversing the graph
for the number $m$ of the edges of the underlying graph.
  Second assign a unique identifier to each party in $O(n)$ rounds
  with 
bit complexity $O(n\log n)$ by
  traversing the spanning tree starting at the leader
(the first and second steps can be merged, but we here describe them separately just for simplicity).
Finally, gather
  into the leader the information of what parties are adjacent to each
  party by conveying adjacency matrices along the spanning tree as follows:
Each party communicates with each adjacent party to know
  the identifier of the adjacent party in one round with $O(m\log n)$ bit complexity.
Next, each leaf node $i$ prepares an
  $n$-by-$n$ adjacency matrix with all entries being zero, puts $1$ in the entries $(i,j)$ 
of the matrix for all adjacent parties $j$, and then sends the matrix to its parent
  node of the tree with $O(n^2)$ bit complexity.  Every internal node $k$ of the tree merges all received matrices,
 puts $1$ in the entries $(k,j)$ for all adjacent parties $j$, and then sends
  the resulting matrix to its parent node.  Finally, the leader can obtain the
  adjacency matrix of the underlying graph, and he then broadcasts
  the matrix along the tree.  These gathering and broadcasting steps take $O(n)$ rounds
  with 
  bit complexity $O(n^3)$. 
\end{proof}

\noindent We now give a proof of Corollary~\ref{cr:Boolean}.

\begin{proofof}{Corollary~\ref{cr:Boolean}}
Once a unique leader is elected and the underlying graph is
recognized, it is sufficient for the leader to gather 
the input bit of every party
with his identifier of $O(\log n)$ bits along the spanning tree.
This input gathering can be done in $O(n)$ rounds with bit
complexity $O(n^2\log n)$.  Thus, together with Corollary~\ref{cr:LE} and
Lemma~\ref{lm:graph-recognition}, 
any computable Boolean function
can be computed in $O(n)$ rounds with bit complexity $O(mn^2)$
for the number $m$ of the edges of the underlying graph.
More generally, suppose that every party $i$ has a qubit so that the
$n$ parties share some $n$-qubit state $\xi$, and let $\rho$ be any
$n$-qubit quantum state computable from $\xi$ and the underlying
graph.  Then, by replacing an input bit with an input qubit for each
party in the above proof for classical case, the leader can gather the
$n$ qubits to have $\xi$ in his local space. Now the leader can
locally generate $\rho$ from $\xi$, and send back the corresponding
qubit to each party to share $\rho$, again along the spanning tree,
in $O(n)$ rounds with $O(n^2 \log n)$ bit complexity.
This completes the
proof of Corollary~\ref{cr:Boolean}. 
\end{proofof}

%%% Local Variables: 
%%% mode: latex
%%% TeX-master: "KobMatTan10"
%%% End: 

\section{GHZ-State Sharing Problem}
\label{sec:ProofOfConsensus}
In this section, we prove Theorem~\ref{th:QuantumConsensus} by
reducing the GHZ-state sharing problem
to computing function $F_k$,
where
$F_k$ is a function such that
$F_k(x_1, \ldots, x_n) = \sum_{i=1}^n x_i \pmod k$
for distributed inputs $x_i\in \set{0,\dots ,k-1}$.  Hereafter, we assume
the existence of an $F_k$-algorithm.
The basic idea can be well understood by considering the case
of $k=2$. 
\subsection{Basic Case ($\boldsymbol{k=2}$)}
The algorithm consists of two phases. The first phase runs 
two attempts of the same procedure
in parallel, each of which lets all parties
share either $(\ket{0}^{\otimes n}\+\ket{1}^{\otimes n})/\sqrt{2}$ or
$(\ket{0}^{\otimes n}-\ket{1}^{\otimes n})/\sqrt{2}$ each with
probability $1/2$.  If the parties share at least one copy of
$(\ket{0}^{\otimes n}\+\ket{1}^{\otimes n})/\sqrt{2}$ after the first
phase, they succeed. If the parties share
two copies of $(\ket{0}^{\otimes n} - \ket{1}^{\otimes n})/\sqrt{2}$,
the second phase distills the
state $(\ket{0}^{\otimes n}\+\ket{1}^{\otimes n})/\sqrt{2}$ 
from them
with
classical communication and partial measurements.  A more detailed
description is as follows.

Let $i\in \set{1,2}$ be the index of each attempt of the procedure performed in the first phase.
The first phase performs the following procedure for each $i$
(notice that function $F_2$ is equivalent to the parity of distributed $n$ bits).

\begin{enumerate}
\item Every party prepares two single-qubit registers $\reg{R}_i$ and $\reg{S}_i$ initialized to $\ket{0}$.
\item Every party applies Hadamard operator 
$\bfH=
\frac{1}{\sqrt{2}}
\big(
\begin{smallmatrix}
  1 & 1\\
1 & -1 
\end{smallmatrix}
\big)
$
 to $\reg{R}_i$:
$ \ket{0}^{\otimes n}\to \frac{1}{\sqrt{2^n}}\sum_{\vec{x}\in \set{0,1}^n} \ket{\vec{x}} $
\item All parties collaborate to 
compute the parity (i.e., the sum modulo $2$) of the contents of $\reg{R}_i$ of all parties
and store the result into $\reg{S}_i$ of each party:
\[ \frac{1}{\sqrt{2^n}}\sum_{\vec{x}\in \set{0,1}^n}
\ket{\vec{x}}\ket{0}^{\otimes n}\to
\frac{1}{\sqrt{2^n}}\sum_{\vec{x}\in \set{0,1}^n}
\ket{\vec{x}}\big|\abs{\vec{x}}\!\!\!\pmod
2\big\rangle^{\otimes n}. \]
\item Every party measures $\reg{S}_i$ in the basis $\set{\ket{0},\ket{1}}$ and 
applies $\bfH$ to $\reg{R}_i$:
\[
\begin{cases}
\displaystyle{\frac{1}{\sqrt{2^{n-1}}}\sum_{\vec{x}\in \set{0,1}^n\colon \abs{\vec{x}}\text{ is even}}} \ket{\vec{x}}
\to 
\frac{\ket{0}^{\otimes n}\+\ket{1}^{\otimes n}}{\sqrt{2}} & \text{if $\ket{0}$ was measured,}\\
&\\
\displaystyle{\frac{1}{\sqrt{2^{n-1}}}\sum_{\vec{x}\in \set{0,1}^n\colon \abs{\vec{x}}\text{ is odd}}} \ket{\vec{x}}
\to 
\frac{\ket{0}^{\otimes n}-\ket{1}^{\otimes n}}{\sqrt{2}} & \text{if $\ket{1}$ was measured.}
\end{cases}
\]
\end{enumerate}
If the state over all $\reg{R}_i$'s is $(\ket{0}^{\otimes n}\+\ket{1}^{\otimes n})/\sqrt{2}$ for at least one of $i=1,2$, we are done; otherwise,
we go on to the second phase. Observe that the state over all $\reg{R}_1$'s and $\reg{R}_2$'s is
\[
(\ket{0}_{\reg{R}_1}^{\otimes n}-\ket{1}_{\reg{R}_1}^{\otimes n})\otimes
(\ket{0}_{\reg{R}_2}^{\otimes n}-\ket{1}_{\reg{R}_2}^{\otimes n})
=(\ket{0}_{\reg{R}_1}^{\otimes n}\ket{0}_{\reg{R}_2}^{\otimes n}
\+
\ket{1}_{\reg{R}_1}^{\otimes n}\ket{1}_{\reg{R}_2}^{\otimes n}
)
-(
\ket{0}_{\reg{R}_1}^{\otimes n}\ket{1}_{\reg{R}_2}^{\otimes n}
\+
\ket{1}_{\reg{R}_1}^{\otimes n}\ket{0}_{\reg{R}_2}^{\otimes n}
),
\]
where we omit 
normalization coefficients.
If every party locally computes the parity of the contents of $\reg{R}_1$'s and $\reg{R}_2$'s and measures the result,
the entire state will be either 
$\ket{0}_{\reg{R}_1}^{\otimes n}\ket{0}_{\reg{R}_2}^{\otimes n}\+\ket{1}_{\reg{R}_1}^{\otimes n}\ket{1}_{\reg{R}_2}^{\otimes n}$ or
$\ket{0}_{\reg{R}_1}^{\otimes n}\ket{1}_{\reg{R}_2}^{\otimes n}\+\ket{1}_{\reg{R}_1}^{\otimes n}\ket{0}_{\reg{R}_2}^{\otimes n}.$
It is easy to see that 
the state $\ket{0}^{\otimes n}\+\ket{1}^{\otimes n}$ can be obtained
from any of these states 
by applying a CNOT to $\reg{R}_2$
using $\reg{R}_1$ as control (all $\reg{R}_2$'s are disentangled).
If we use a quantum simulation of a classical algorithm
that deterministically computes the parity of distributed $n$ bits (e.g., view-based algorithms~\cite{YamKam96IEEETPDS-1,KraKriBer94InfoComp,TanKobMat05STACS}),
our algorithm uses only a constant-sized gate set.

\subsection{General Case ($\boldsymbol{k>2}$)}
In the following, we assume $k$-level qudits are available 
for simplicity (the algorithm can easily be carried over the case where we are allowed to use only qubits).
Any pure state of a $k$-level
qudit can be represented as $\sum
_{i=0}^{k-1} \alpha_i\ket{i}$ with complex numbers $\alpha_i$ such that
$\sum _{i=0}^{k-1}\abs{\alpha_i}^2=1$ (for $k=2$, this is just a
qubit).

Our algorithm uses the following operator $\bfW _k$ over one $k$-level qudit, instead of $\bfH$
used in the case of $k=2$: For $x\in \set{0,\dots ,k\!\-1}$,
\[ \bfW _k \ket{x}=\frac{1}{\sqrt{k}}\sum _{j=0}^{k-1}\omega _k ^{x j}\ket{j},\]
where $\omega _k = e^{\frac{2\pi}{k}i}$. 
In what follows,
we denote $\left(\sum _{x=0}^{k-1}\omega_k ^{t\cdot x}\ket{x}^{\otimes n}\right)/{\sqrt{k}}$ by $\cat{k}{t}$.
For instance, $\cat{2}{0}$ denotes ${(\ket{0}^{\otimes n}+\ket{1}^{\otimes n})}/{\sqrt{2}}$.

\subsubsection{First Phase}
The first phase is for the purpose of 
sharing $k$ states drawn from the set 
$
 \left\{
\cat{k}{t}
 \colon t\in \set{0,\dots , k-1}\right
\}.
$
The operations are described as follows, which are similar to the case of $k=2$.
\begin{algorithm*}{First Phase}
  \begin{flushleft}
\hspace{10pt} For $i=1,2,\dots, k$, perform the following operations in parallel:
  \end{flushleft}
  \begin{step}
\setlength{\itemsep}{0pt}
  \item Prepare a single-qudit register $\reg{R}_i$ initialized to $\ket{0}$.
  \item Apply $\bfW _k$ to the qudit in $\reg{R}_i$, which maps the  state $\ket{0}^{\otimes n}$
into $\left( \bfW _k \ket{0}\right)^{\otimes n}=
\frac{1}{\sqrt{k^n}} \sum _{y=0}^{k^n-1}\ket{y}.$
  \item Run an $F_k$-algorithm to compute the value of 
$F_k(y):=\sum _{j=1}^{n}y_j\pmod k$, where $y_j$ is the content of $\reg{R}_i$ of the $j$th party, 
and store the result into a single-qudit register $\reg{S}_i$:
\[
\frac{1}{\sqrt{k^n}} \sum _{y=0}^{k^n-1}\ket{y}\ket{0}^{\otimes n}\to
\frac{1}{\sqrt{k^n}} \sum _{y=0}^{k^n-1}\ket{y}\ket{F_k(y)}^{\otimes n}.
\]
  \item Measure $\reg{S}_i$ in the  basis $\set{\ket{0},\dots, \ket{k-1}}$.
If $s_i\in \set{\ket{0},\dots ,\ket{k-1}}$ is measured, the state is 
\[
\frac{1}{\sqrt{k^{n-1}}} \sum _{
\stackrel{\scriptstyle y\in \set{0,\dots ,k^n-1}\colon}
{\scriptstyle \sum _{j=1}^{n} y_j=s_i\!\!\!\!\pmod{k}}}\ket{y}.
\]
  \item Apply $\bfW _k^{\dagger}$ to $\reg{R}_i$.
  \end{step}
\end{algorithm*}

The following lemma implies that, for each $i\in \set{1,\dots ,k}$,
the state of $\reg{R}_i$'s after the first phase is 
$\cat{k}{-s_i\bmod k}$.
If $s_i=0$ for some $i$, we are done.
Otherwise, the parties perform the second phase (described later)
to distill the state $\cat{k}{0}$
from the $k$ states shared by all parties.

\begin{lemma}
\label{lm:WAppliedToCat}
  \begin{equation*}
 \bfW _k ^{\otimes n}\left(\frac{1}{\sqrt{k}} \sum _{x=0}^{k-1}\omega_k ^{t\cdot x}\ket{x}^{\otimes n} \right)
=\frac{1}{\sqrt{k^{n-1}}}\sum_{
\stackrel{\scriptstyle y\in \set{0,\dots ,k^n-1}\colon}
{\scriptstyle t+\sum _{j=1}^{n} y_j=0\pmod k}} \ket{y}.
  \end{equation*}
\end{lemma}
The proof is given in Appendix.

\subsubsection{Second Phase}
Suppose that, after the first phase, all parties share $k$ states,
$\cat{k}{-s_i \bmod k}$ with $s_i\neq 0$ for $i=1,\dots ,k$.
Then, there must be two integers $l,m\in \set{1,\dots ,k}$ with $s_l=s_m$, since $s_l,s_m \in \set{1,\dots , k-1}$.
We can distill the state $\cat{k}{0}$ from the states
$\cat{k}{-s_l \bmod k}$ and $\cat{k}{-s_m \bmod k}$ as follows.

Suppose that $n$ parties share two copies of $\cat{k}{t}$ for any $t\in \set{1,\dots ,k-1}$ for their quantum registers
$\reg{R}_1$'s and $\reg{R}_2$'s. Namely, the state over all $\reg{R}_1$'s and $\reg{R}_2$'s is
\[
    \bigg( \underbrace{\frac{1}{\sqrt{k}}\sum _{x=0}^{k-1}\omega_k ^{t\cdot x}\ket{x}^{\otimes n}}_{\reg{R}_1\text{'s}}\bigg)\otimes
   \bigg( \underbrace{\frac{1}{\sqrt{k}}\sum _{x=0}^{k-1}\omega_k ^{t\cdot x}\ket{x}^{\otimes n}}_{\reg{R}_2\text{'s}}\bigg).\]
By rearranging the registers, the state is
\[
\frac{1}{k}
\sum _{r=0}^{k-1}\sum _{x=0}^{k-1}\omega _k^{t\cdot r}
\left(\ket{x}_{\reg{R}_1}\ket{r-x\!\!\!\!\pmod{k}}_{\reg{R}_2}\right)^{\otimes n}.
\] 
Every party then performs the following operations.
  \begin{algorithm*}{Second Phase}
    \begin{step}
    \item Add the content of $\reg{R}_1$ to the content of
      $\reg{R}_2$ under modulo $k$: The state becomes
\[
\frac{1}{k}
\sum _{r=0}^{k-1}\sum _{x=0}^{k-1}\omega _k^{t\cdot r}
\left(\ket{x}_{\reg{R}_1}\
\ket{r\bmod k}_{\reg{R}_2}
\right)^{\otimes n}.
\] 
    \item 
Measure $\reg{R}_2$ in the basis $\set{\ket{0},\dots ,\ket{k-1}}$ and let $r$ be the measurement result: The state is
\[
\frac{\omega _k^{t\cdot r}}{\sqrt{k}}
\sum _{x=0}^{k-1}
\ket{x}_{\reg{R}_1}
^{\otimes n}.
\] 
    \item Output $\reg{R}_1$.
    \end{step}
  \end{algorithm*}
\subsubsection{Proof of Theorem~\ref{th:QuantumConsensus}}
The correctness of the algorithm follows from the above description of the algorithm. The communication occurs only when computing $F_k$ (in the
first phase). Thus, the algorithm works in
$O(Q^{\operatorname{rnd}}(\calF_k))$ 
rounds with bit complexity 
$O(Q^{\operatorname{bit}}(\calF_k))$,
where $\calF_k$ is the given $F_k$-algorithm.
The algorithm works with  the operators $\bfW _k, \bfW _k^{\dagger}$,
the operators for computing classical functions (such as addition under modulo $k$) that are independent of $n$,
except the given $F_k$-algorithm. Therefore, the algorithm can be implemented with a gate set
whose size is finite and independent of $n$ if an $F_k$-algorithm is given. \hfill \qed

%%% Local Variables: 
%%% mode: latex
%%% TeX-master: "KobMatTan10"
%%% End: 

\section{Conclusion}
We proved that the leader election problem $\LE_n$ can exactly be
solved with at most the same complexity (up to constant factor) as
that of computing symmetric Boolean functions on an anonymous quantum
network. In particular, the hardness of the leader election problem is
characterized by that of computing $H_0$, a function of checking if
all parties each have the bit $0$. This shows that quantum information
can change the hardness relation among distributed computing problems
(recall that $H_0$ can be computed for all network topologies but $\LE_n$
cannot, on anonymous classical networks).

In the proof, we used (given) distributed algorithms for computing
symmetric Boolean functions to implement phase-shift operators of
amplitude amplification. Here, assuming that the underlying graph is
undirected, we were able to erase the garbage left by the algorithms
by the standard technique of inverting all operations and
communications performed.
We do not know if our proof works (with
modifications) even on directed networks.  It is also a open question as to whether
$\LE_n$ can exactly be solved in rounds linear in the number of
parties when the underlying graph is directed (notice that the leader
election algorithms in Ref.~\cite{TanKobMat05STACS} work with some
modifications even on directed networks,  but they require rounds super-linear in the number of parties).

We also gave an quantum algorithm that exactly solves the GHZ-state sharing problem
in rounds linear in the number of parties with a constant-sized gate set, if the network is undirected.
It is still open whether the problem can exactly be solved in linear rounds on directed networks.
If much more rounds are allowed, we can solve the problem on directed networks by modifying the idea
in Ref.~\cite{TanKobMat05STACS}.

%%% Local Variables: 
%%% mode: latex
%%% TeX-master: "KobMatTan10"
%%% End: 

\newpage \appendix
\section*{Appendix}
\section{GHZ-State Sharing Problem}
\label{appdx:QuantumConsensusProblem}
\begin{proofof}{Lemma~\ref{lm:WAppliedToCat}}
We first prove the lemma for $t=0$. The proof can easily be generalized to the case of $t>0$.
Notice that
\begin{equation}
\label{eq:lm:WAppliedToCat}
  \bfW _k ^{\otimes n}\left(\frac{1}{\sqrt{k}} \sum _{x=0}^{k-1}\ket{x}^{\otimes n}\right)
=\frac{1}{\sqrt{k^{n+1}}}\sum _{y=0}^{k^n-1}\alpha_y\ket{y}
= \frac{1}{\sqrt{k^{n+1}}}\sum _{y=0}^{k^n-1} (\alpha_y^{(0)}+\dots +\alpha_y^{(k-1)})\ket{y},
\end{equation}
where $\alpha_y:=\sum _{x=0}^{k-1}\alpha_y^{(x)}$ and  $\alpha_y^{(x)}/\sqrt{k^n}=\bra{y}\bfW _k^{\otimes n}\ket{x}^{\otimes n}$.
Let $y=y_1y_2\dots y_n$ for  $y_j\in \set{0,\dots ,k-1}$.
By the definition of $\bfW _k$, we have
\begin{equation*}
  \alpha_y^{(x)}=\prod _{j=1}^{n}\omega_k^{x\cdot y_j}=(\alpha_y^{(1)})^x\hspace{3mm}\text{(for $0\le x\le k-1$)}.
\end{equation*}
Therefore, the following claim holds.

\begin{claim}
\begin{equation*}
  \alpha_y = \sum _{x=0}^{k-1}\alpha_y^{(x)}=\sum _{x=0}^{k-1}(\alpha_y^{(1)})^x,
\end{equation*}
where $\alpha_y^{(1)}\in \set{1,\omega_k^1,\omega_k^2, \dots ,\omega_k^{k-1}}$.
\end{claim}
We next calculate $\alpha_y$ for each $y$.
If $\alpha_y^{(1)}=1$, then $\alpha_y=\sum _{x=0}^{k-1}(\alpha_y^{(1)})^x=k$ by the above claim.
If $\alpha_y^{(1)}=\omega _k^p$ for some number $p$ prime to $k$,
then 
\[ \alpha_y=\sum _{x=0}^{k-1}\omega_k^{px}=0.  \]
Suppose that $\alpha_y^{(1)}=\omega _k^q$ for some number $q$ not prime to $k$.
Let $g$ be the greatest common divisor (GCD) of $q$ and $k$. Since $\alpha_y^{(1)}=e^{2\pi\frac{q/g}{k/g}i}$ 
is the $(k/g)$th root of 1, we have
$\sum _{j=0}^{k/g-1}(\alpha_y^{(1)})^j=0.$
Therefore, 
\[
\alpha_y=\sum _{x=0}^{k-1}(\alpha_y^{(1)})^x=\sum _{m=1}^{g}\sum _{j=0}^{k/g-1}(\alpha_y^{(1)})^j=0.
\]
Hence, only the basis states $\ket{y}$ such that $\alpha_y^{(1)}=1$ have non-zero amplitudes.
Since $\alpha_y^{(1)}=\prod _{j=1}^{n}\omega_k^{y_j}=\omega_k^{\sum_{j=1}^{n}y_j}$,
\[
\sum _{y=0}^{k^n-1}\alpha_y\ket{y}
= 
\sum_{
\stackrel{\scriptstyle y\in \set{0,\dots ,k^n-1}\colon}
{\scriptstyle \sum _{j=1}^{n} y_j=0\pmod k}} k\ket{y}.
\]
Thus, the lemma for $t=0$ follows  from eq. \eqref{eq:lm:WAppliedToCat}.

We now consider the case of $t>0$.
Suppose that
\begin{equation*}
  \bfW _k^{\otimes n}\left(\frac{1}{\sqrt{k}}\sum _{x=0}^{k-1}\omega_k ^{t\cdot x}\ket{x}^{\otimes n}\right)
=\frac{1}{\sqrt{k^{n+1}}}\sum _{y=0}^{k^n-1}\beta_y\ket{y}
= \frac{1}{\sqrt{k^{n+1}}}\sum _{y=0}^{k^n-1} (\beta_y^{(0)}+\dots +\beta_y^{(k-1)})\ket{y},
\end{equation*}
where $\beta_y=\sum _{x=0}^{k-1}\beta_y^{(x)}$ and  $\beta_y^{(x)}/{\sqrt{k^n}}=\bra{y}\bfW _k^{\otimes n}(\omega_k ^{t\cdot x}\ket{x}^{\otimes n})$. 
Then, we have  $\beta_y^{(x)}=\omega_k ^{t\cdot x}\alpha_y^{(x)}=(\omega_k ^{t}\alpha_y^{(1)})^x.$ This implies that
\[  \beta_y = \sum _{x=0}^{k-1}\beta_y^{(x)}=\sum _{x=0}^{k-1}(\beta_y^{(1)})^x. \]
By an argument similar to the case of $t=0$,
only the basis states $\ket{y}$ such that $\beta_y^{(1)}=1$ have non-zero amplitudes.
The lemma follows from   
$\beta_y^{(1)}=\omega_k^{t+\sum_{j=1}^{n}y_j}$.
\end{proofof}

%%% Local Variables: 
%%% mode: latex
%%% TeX-master: "KobMatTan09"
%%% End: 


\begin{thebibliography}{10}

\bibitem{afek-matias94}
Yehuda Afek and Yossi Matias.
\newblock Elections in anonymous networks.
\newblock {\em Information Computation}, 113(2):312--330, 1994.

\bibitem{Ang80STOC}
Dana Angluin.
\newblock Local and global properties in networks of processors (extended
  abstract).
\newblock In {\em Proceedings of the Twentieth Annaul ACM Symposium on Theory
  of Computing}, pages 82--93, 1980.

\bibitem{AttSniWar88JACM}
Hagit Attiya, Marc Snir, and Manfred~K. Warmuth.
\newblock Computing on an anonymous ring.
\newblock {\em Journal of the ACM}, 35(4):845--875, 1988.

\bibitem{BraHoyMosTap02AMS}
Gilles Brassard, Peter H{\o}yer, Michele Mosca, and Alain Tapp.
\newblock Quantum amplitude amplification and estimation.
\newblock In {\em Quantum Computation and Quantum Information: A Millennium
  Volume}, volume 305 of {\em AMS Contemporary Mathematics Series}, pages
  53--74. AMS, 2002.

\bibitem{chi-kim98}
Dong~Pyo Chi and Jinsoo Kim.
\newblock Quantum database search by a single query.
\newblock In {\em Proceedings of the First NASA Int. Conf. Quantum Computing
  and Quantum Communications (QCQC)}, volume 1509 of {\em LNCS}, pages
  148--151. Springer, 1998.

\bibitem{DenPan06SIGACT}
Vasil~S. Denchev and Gopal Pandurangan.
\newblock Distributed quantum computing: A new frontier in distributed systems
  or science fiction?
\newblock {\em ACM SIGACT News}, 39(3):77--95, 2006.

\bibitem{DHoPan06QIC}
Ellie D'Hondt and Prakash Panangaden.
\newblock The computational power of the {W} and {GHZ} states.
\newblock {\em Quantum Information and Computation}, 6(2):173--183, 2006.

\bibitem{GavKosMar09ARXIV}
Cyril Gavoille, Adrian Kosowski, and Marcin Markiewicz.
\newblock What can be observed locally? {R}ound-based models for quantum
  distributed computing.
\newblock arXiv:0903.1133, 2009.

\bibitem{ItaRod81FOCS}
Alon Itai and Michael Rodeh.
\newblock Symmetry breaking in distributive networks.
\newblock In {\em Proceedings of the Twenty-Second Annual IEEE Symposium on
  Foundations of Computer Science}, pages 150--158, 1981.

\bibitem{ItaRod90InfoComp}
Alon Itai and Michael Rodeh.
\newblock Symmetry breaking in distributed networks.
\newblock {\em Information Computation}, 88(1):60--87, 1990.

\bibitem{KitSheVya02Book}
Alexei~Yu. Kitaev, Alexander~H. Shen, and Mikhail~N. Vyalyi.
\newblock {\em Classical and Quantum Computation}, volume~47 of {\em Graduate
  Studies in Mathematics}.
\newblock AMS, 2002.

\bibitem{KraKri92SPDP}
Evangelos Kranakis and Danny Krizanc.
\newblock Distributed computing on cayley networks (extended abstract).
\newblock In {\em Proceedings of the Fourth IEEE Symposium on Parallel and
  Distributed Processing}, pages 222--229. IEEE Computer Society Press, 1992.

\bibitem{KraKri97JALG}
Evangelos Kranakis and Danny Krizanc.
\newblock Distributed computing on anonymous hypercube networks.
\newblock {\em Journal of Algorithms}, 23(1):32--50, 1997.

\bibitem{KraKriBer94InfoComp}
Evangelos Kranakis, Danny Krizanc, and Jacov van~den Berg.
\newblock Computing boolean functions on anonymous networks.
\newblock {\em Information Computation}, 114(2):214--236, 1994.

\bibitem{Lyn96Book}
Nancy~A. Lynch.
\newblock {\em Distributed Algorithms}.
\newblock Morgan Kaufman Publishers, 1996.

\bibitem{NieChu00Book}
Michael~A. Nielsen and Isaac~L. Chuang.
\newblock {\em Quantum Computation and Quantum Information}.
\newblock Cambridge University Press, 2000.

\bibitem{PalSinKum03ARXIV}
Sudebkumar~Prasant Pal, Sudhir~Kumar Singh, and Somesh Kumar.
\newblock Multi-partite quantum entanglement versus randomization: Fair and
  unbiased leader election in networks.
\newblock quant-ph:/0306195, 2003.

\bibitem{TanKobMat05STACS}
Seiichiro Tani, Hirotada Kobayashi, and Keiji Matsumoto.
\newblock Exact quantum algorithms for the leader election problem.
\newblock In {\em Proceedings of the Twenty-Second Symposium on Theoretical
  Aspects of Computer Science (STACS 2005)}, volume 3404 of {\em LNCS}, pages
  581--592. Springer, 2005.
\newblock (Full version in {\tt http://jp.arxiv.org/abs/0712.4213)}.

\bibitem{YamKam88PODC}
Masafumi Yamashita and Tsunehiko Kameda.
\newblock Computing on an anonymous network.
\newblock In {\em Proceedings of the Seventh ACM Symposium on Principles of
  Distributed Computing}, pages 117--130, 1988.

\bibitem{YamKam96IEEETPDS-1}
Masafumi Yamashita and Tsunehiko Kameda.
\newblock Computing on anonymous networks: Part {I} -- characterizing the
  solvable cases.
\newblock {\em IEEE Transactions on Parallel Distributed Systems}, 7(1):69--89,
  1996.

\bibitem{YamKam96IEEETPDS-2}
Masafumi Yamashita and Tsunehiko Kameda.
\newblock Computing on anonymous networks: Part {II} -- decision and membership
  problems.
\newblock {\em IEEE Transactions on Parallel Distributed Systems}, 7(1):90--96,
  1996.

\bibitem{yamashita-kameda99}
Masafumi Yamashita and Tsunehiko Kameda.
\newblock Leader election problem on networks in which processor identity
  numbers are not distinct.
\newblock {\em IEEE Transactions on Parallel and Distributed Systems},
  10(9):878--887, 1999.

\end{thebibliography}
\end{document}